\documentclass[aip,jmp,amsmath,amssymb,amsfonts,12pt]{revtex4-2}
\setcitestyle{numbers,square}
\usepackage{bm}
\usepackage{hyperref}
\usepackage{amsthm}
\usepackage{listings}
\usepackage{xcolor}
\usepackage{mathdots}

\newtheorem{theorem}{{Theorem}}[section] 
\newtheorem{defi}{{Definition}}[section]
\newtheorem{proposition}{Proposition}[section]
\newtheorem{coro}[theorem]{Corollary}
\newtheorem{lemme}{{Lemma}}[section]

\begin{document}

\title{Hamiltonian hydrodynamic reductions of one-dimensional Vlasov equations}

\author{Rayan Oufar}
\email{rayan.oufar@univ-amu.fr}

\author{Cristel Chandre}
\email{cristel.chandre@cnrs.fr}

\affiliation{CNRS, Aix Marseille Univ, I2M, Marseille, France}

\begin{abstract}
    We investigate Hamiltonian fluid reductions of the one-dimensional Vlasov-Poisson equation. Our approach utilizes the hydrodynamic Poisson bracket framework, which allows us to systematically identify fundamental normal variables derived from the analysis of the Casimir invariants of the resulting Poisson bracket. This framework is then applied to analyze several well-established Hamiltonian closures of the one-dimensional Vlasov equation, including the multi-delta distribution and the waterbag models. Our key finding is that all of these seemingly distinct closures consistently lead to the formulation of a unified form of parametric closures: When expressed in terms of the identified normal variables, the parameterization across all these closures is revealed to be polynomial and of the same degree. All these parametric closures are uniquely generated from one of the moments, called $\mu_2$, a cubic polynomial in the normal variables. This result establishes a structural connection between these different physical models, offering a path toward a more unified and simplified description of the one-dimensional Vlasov-Poisson dynamics through its reduced hydrodynamic forms with an arbitrary number of fluid variables.  
\end{abstract}

\maketitle

\section{Introduction}

The one-dimensional Vlasov--Poisson system describes the dynamics of a collisionless plasma through its phase--space distribution function \(f(t,x,p)\), which gives the density of particles at time \(t\), position \(x\), and momentum \(p\). It is written as
\begin{equation}
    \label{eq:vp-vlasov}
    \frac{\partial f}{\partial t}=-v \frac{\partial f}{\partial x} - q E(t,x) \frac{\partial f}{\partial p},
\end{equation}
where \(v = p/m\) denotes the particle velocity and $q$ is the electric charge of the particle species. The self-consistent electrostatic field \(E(t,x)\) follows from Poisson’s equation
\begin{equation}
    \frac{\partial E}{\partial x} = q \left(\int f(t,x,p)\,\mathrm{d}p -n_0\right),
\label{eq:vp-poisson}
\end{equation}
with \(n_0\) denoting the density of a uniform neutralizing background that ensures overall charge neutrality. In what follows, we assume $q=m=1$. Although the kinetic description \eqref{eq:vp-vlasov}--\eqref{eq:vp-poisson} is complete, it evolves a function over the full phase space \((x,p)\), which can be computationally and analytically demanding, especially in higher dimensions, and provide more information than needed. A standard approach toward model reduction is to introduce fluid moments of the distribution $f$
\begin{equation}
    \label{eqn:moments}
    P_n(t,x) = \int p^n\, f(t,x,p)\,\mathrm{d}p,
\end{equation}
for $n=0,1,\ldots$, under the assumption that \(f\) decays sufficiently rapidly as \(|p|\to\infty\). These moments correspond to macroscopic plasma quantities: \(P_0\) is the particle density, \(P_1\) the momentum density, \(P_2\) is related to pressure and temperature, and \(P_3\) to the heat flux. By evolving a reduced set of variables in configuration space, one obtains a more tractable model with clearer and more intuitive physical interpretations.

Each dynamical equation for $P_n$, e.g., mass conservation, momentum balance, energy balance, involves a higher-order moment \(P_{n+1}\). The term in the Vlasov equation responsible for this shift in moments is the streaming term $-v~\partial f/\partial x$ (see also Ref.~\cite{Hazeltine2004}).  Thus, the hierarchy does not naturally close and, without additional modeling assumptions, remains untractable as an infinite ladder of partial differential equations.

To obtain a closed fluid model, one must prescribe a closure relation expressing higher-order moments in terms of lower-order ones~\cite{Braginskii_1965,Ott_1969,Hazeltine_1985,Hammett_1990,Hammett_1993,Sugama_2003,Passot_2004,Shadwick_2004,Shadwick_2005,Goswami_2005,Sarazin_2009,Shadwick_2012}. Common closure strategies include:
\begin{itemize}
    \item Isothermal closure: Pressure proportional to density, modeling a fixed-temperature regime.
    \item Adiabatic closure: A polytropic relation such as \(P_2 \propto P_0^\gamma\), appropriate for collisionless compressions.
    \item Cold plasma closure: Neglecting higher-order moments (\(P_{n\ge 2}=0\)), suitable for mono-energetic beams.
    \item Kinetic-based closures: Derived via asymptotic methods such as Chapman--Enskog expansions for weakly collisional regimes.
    \item Closures based on an ansatz on the shape of the distribution function $f$, such as the waterbags~\cite{depackh_water-bag_1962,bertrand_non_1968}, the multi-delta distributions~\cite{Gosse2003,fox_higher-order_2009}, or a generalized Maxwellian distribution~\cite{Shadwick_2004,de_guillebon_hamiltonian_2012}.  
    \item Hamiltonian closures~\cite{perin_higher-order_2014-1,perin_hamiltonian_2015-2,chandre_four-field_2022}: Ensuring the Hamiltonian structure for the truncated system, retaining the Hamiltonian character of the Vlasov equation.
\end{itemize}
The chosen closure controls the level of kinetic information retained and must balance physical accuracy, mathematical well-posedness, and computational feasibility. When selected appropriately, moment closures capture important macroscopic behavior---including wave propagation and pressure effects---while reducing the complexity inherent to full kinetic simulations. Usually these closures amount to finding suitable functions, aka, equations of state, of the form
\begin{equation}
    \label{eqn:state}
    P_{n\geq N}=P_{n\geq N}(P_0,P_1,\ldots,P_{N-1}),
\end{equation}
in order to truncate the infinite ladder to the first $N$ fluid moments. Such equations of state are often selected for their practicality and their modeling of specific conditions of the plasma. Given the explicit form of the bracket, $N-3$ functions $P_{n\geq N}$ of $N$ variables should be specified. 

In Refs.~\cite{perin_higher-order_2014-1,perin_hamiltonian_2015-2,chandre_four-field_2022}, a strategy to find relevant closures was developed based on the Hamiltonian structure of the parent kinetic equation. By solving directly the set of nonlinear partial differential equations for the closure functions resulting from the Jacobi identity, Hamiltonian closure functions were identified and studied for $N\leq 4$. In particular, in Ref.~\cite{chandre_four-field_2022}, it was shown that these closure functions for $N=4$ can take very intricate forms, and sometimes cannot be formulated in an explicit way. Instead, it was shown that a parametric expression of the same cumbersome closure can express it in a quite simple form. These parametric forms for the moments are
$$
    P_n = P_n(\nu_1,\ldots,\nu_N), 
$$
for $n\geq 0$, and with as many carefully selected parameters $\nu_i$ as the number of field variables. In particular, for $n=0,\ldots,N-1$, the first $N$ relations can be locally inverted $\nu_i=\nu_i(P_0,\ldots,P_{N-1})$ and inserted in the remaining $N-3$ equations $P_{n\geq N}$ to define (locally) an equation of state of the form~\eqref{eqn:state}. Of course, this inversion should always be possible in principle (to ensure the Jacobi identity), but often cannot be written explicitly. 
In the following, we demonstrate that defining closures parametrically significantly simplifies the reduction: We show that the closure is generated by a single function $\mu_2$ of $N-2$ variables, instead of the initial $N-3$ functions $P_{n\geq N}$ of $N$ variables. Furthermore, by examining all known Hamiltonian closures of the one-dimensional Vlasov-Poisson equation, we show that they each admit a polynomial parameterization in $\nu_k$ uniquely determined by a single cubic polynomial $\mu_2$ in $N-2$ normal variables. This underlying structure reveals a deep connection between seemingly distinct closures, providing a unified framework for describing the one-dimensional Vlasov-Poisson dynamics through reduced hydrodynamic models with an arbitrary number of fluid variables.

In Sec.~\ref{sec:part0}, we recall some known elements on the Hamiltonian formulation of the Vlasov equation and its fluid reductions.
In Sec.~\ref{sec:part1}, we define what hydrodynamic Poisson brackets for fluid and plasmas are, and provide some general properties. In Sec.~\ref{sec:part2}, we investigate all known Hamiltonian closures in the light of the properties of the hydrodynamic brackets. 

\section{Preliminaries}
\label{sec:part0}

The Vlasov equation possesses a Hamiltonian structure~\cite{morrison_poisson_1982,morrison_hamiltonian_1998}, in the sense that it can be derived from a Poisson bracket $\{\cdot,\cdot\}$ and a Hamiltonian $H$ as $\dot{f}=\{f,H\}$.

We remind that a \emph{Poisson bracket} on an algebra $\mathcal{A}$ is a bilinear operator $\{\cdot, \cdot\} : \mathcal{A} \times \mathcal{A} \to \mathcal{A}$
satisfying the following properties for all $F, G, H \in \mathcal{A}$:\\
- antisymmetry: $\{F, G\} = -\{G, F\}$, \\
- Jacobi identity: $\{F, \{G, H\}\} + \{G, \{H, F\}\} + \{H, \{F, G\}\} = 0$, \\
- Leibniz rule: $\{F, GH\} = \{F, G\} H + G \{F, H\}$.

The Hamiltonian associated with the Vlasov-Poisson equation is defined as
$$
    H[f]=\iint f\frac{p^2}{2}\ {\rm d} x {\rm d} p + \frac{1}{2}\int E^2 \ {\rm d} x,
$$
where $E[f]$ is given by Eq.~\eqref{eq:vp-poisson} and the Poisson bracket
\begin{equation}
\label{eqn:PB_Vlasov}
     \{F,G\} = \iint f\left( \frac{\partial}{\partial x}\frac{\delta F} {\delta f} \frac{\partial}{\partial p}\frac{\delta G}{\delta f} - \frac{\partial}{\partial p}\frac{\delta F} {\delta f} \frac{\partial}{\partial x}\frac{\delta G}{\delta f}\right) \ {\rm d} x {\rm d} p.
\end{equation}
Using the moments~\eqref{eqn:moments}, the Poisson bracket reduces to the Kupershmidt–Manin bracket~\cite{kupershmidt_equations_1978,kupershmidt_hydrodynamical_1987}
\begin{equation}
\label{eqn:PB_KM}
    \{F,G\} = \int m P_{n+m-1} \left( \partial_x F_n G_m -F_m \partial_x G_n\right) \ {\rm d} x,
\end{equation}
with implicit summation over repeated indices $n$ and $m$, and where $F_n$ denotes the functional derivative of the observable $F$ with respect to $P_n$. Using an integration by parts, this Poisson bracket is rewritten as
\begin{equation}
    \label{eqn:PB_P}
    \{F,G\} = \int \left(\partial_x F_n \alpha_{nm}({\bf P}) G_m + F_n \beta_{nm}({\bf P}, \partial_x {\bf P}) G_m\right) \ {\rm d} x,
\end{equation}
where 
\begin{subequations}
\label{eqn:albeP}
    \begin{eqnarray}
        && \alpha_{nm}({\bf P}) = (n+m) P_{n+m-1},\\
        && \beta_{nm}({\bf P}, \partial_x {\bf P}) = n \partial_x P_{n+m-1}.
    \end{eqnarray}
\end{subequations}

We notice that $\alpha$ is symmetric and $\partial_x \alpha = \beta + \beta^t$, to ensure that the resulting bracket is antisymmetric.
If the change of variables ${\bf P}\mapsto {\bf Q}={\bf Q} ({\bf P})$, i.e., a change of variables which only depend on the values of the original field variables (and not, e.g., of its derivatives) is performed, the form~\eqref{eqn:PB_P} of the Poisson bracket remains unchanged with $\alpha$ and $\beta$ replaced by $\overline{\alpha}$ and $\overline{\beta}$ given by
\begin{subequations}
\label{eqn:change}
\begin{eqnarray}
    && \overline{\alpha}_{kl}({\bf Q}) = \frac{\partial Q_k}{\partial P_n}\alpha_{nm}\frac{\partial Q_l}{\partial P_m},\label{eqn:change_1} \\
    && \overline{\beta}_{kl}({\bf Q}, \partial_x {\bf Q}) = \partial_x \frac{\partial Q_k}{\partial P_n}\alpha_{nm}\frac{\partial Q_l}{\partial P_m}+\frac{\partial Q_k}{\partial P_n}\beta_{nm}\frac{\partial Q_l}{\partial P_m}.
\end{eqnarray}
\end{subequations}
In particular, we notice that $\overline{\alpha}$ is still a function of the values of the field variables (and not of their derivatives), and that if $\beta$ is linear in $\partial_x {\bf P}$, then $\overline{\beta}$ is a function of ${\bf Q}$ and $\partial_x {\bf Q}$, linear in $\partial_x {\bf Q}$. We also notice that the matrix $\overline{\alpha}$ is still symmetric and that it satisfies $\partial_x \overline{\alpha} = \overline{\beta} + \overline{\beta}^t$. The invariance of the Poisson brackets of the form~\eqref{eqn:PB_P} led to the definition of hydrodynamic brackets~\cite{dubrovin_hamiltonian_1983,mokhov_dubrovin-novikov_1988,dubrovin_hydrodynamics_1989,mokhov_classification_2008}. Given this invariance, we explore the Hamiltonian closures within this family of hydrodynamic Poisson brackets. 

Truncating Poisson brackets of the form~\eqref{eqn:PB_P}, i.e., restricting the sum over $n$ and $m$ from 0 to $N-1$, and imposing a closure of the form~\eqref{eqn:state}, still preserve the shape of the bracket, in the sense that $\alpha$ remains only a function of the values of the field variables, and $\beta$ remains a function of the values and first derivatives of these field variables, linear in the first derivatives. However, the Jacobi identity identity is most likely broken for a generic equation of state of the form~\eqref{eqn:state}. As a result, the dynamics generated by truncated brackets of the form~\eqref{eqn:PB_P} is no longer Hamiltonian. The central question in finding Hamiltonian closures if for which Eq.~\eqref{eqn:state}, the bracket~\eqref{eqn:PB_P} is a Poisson bracket? 

\section{Hydrodynamic Poisson brackets for fluid and plasmas}
\label{sec:part1}

Here we consider the algebra $\mathcal{A}$ of functionals $F$ of the $N$ fields ${\bf u}(x)=(u_1(x),u_2(x),\ldots,u_N(x))$. In what follows, we assume that the field variables $u_n$, and all their derivatives vanish at infinity. We also assume implicit summation over repeated indices.

\subsection{Definition and properties of hydrodynamic Poisson brackets}

Following the above discussion and in particular Eq.~\eqref{eqn:PB_P}, we define a hydrodynamic bracket~\cite{dubrovin_hamiltonian_1983,dubrovin_hydrodynamics_1989,mokhov_dubrovin-novikov_1988} as
\begin{equation}
    \label{eqn:PBh}
    \{F, G\} = \int \left( \partial_x F_n \alpha_{nm}({\bf u}) G_m + F_n \beta_{nm}({\bf u}, \partial_x {\bf u}) G_m\right) {\rm d}x,
\end{equation}
where $F_n={\delta F}/{\delta u_n}$ is the functional derivatives of $F$ with respect to the field variable $u_n$, and the implicit summation over $n$ and $m$ is from 1 to $N$.
For these brackets, $\alpha({\bf u})$ is symmetric, i.e., $\alpha=\alpha^t$, and only depends on the values of the functions $u_n(x)$ and $\beta_{nm}({\bf u},\partial_x {\bf u}) = \beta_{nmk}({\bf u})\partial_x u_k$, for all $n,m=1,\ldots, N$, i.e., $\beta$ is linear in $\partial_x {\bf u}$. These brackets are antisymmetric if and only if
$$
    \partial_x \alpha_{nm} = \beta_{nm} + \beta_{mn},
$$
and they obviously satisfy the Leibniz rule.  The truncated bracket~\eqref{eqn:PB_P} in the moments $(P_0,\ldots,P_{N-1})$ is a hydrodynamic bracket for all $N$. 

{\em Remark:} Using an integration by parts, a hydrodynamic bracket can be mapped to:
\begin{equation}
    \label{eqn:PBa}
    \{F,G\} = \int a_{nm}({\bf u})\left( F_n\partial_x G_m - G_n\partial_x F_m \right)\, {\rm d}x, 
\end{equation}
where $\alpha_{nm}= a_{nm} + a_{mn}$ and $\beta_{nm}= \partial_x a_{mn}$, as it is for the bracket~\eqref{eqn:PB_KM}. The inconvenience of working with the brackets of the form~\eqref{eqn:PBa} is the presence of a gauge invariance: By adding a constant antisymmetric matrix to $a$, the expression of the bracket, and in particular, $\alpha$ and $\beta$, does not change. 

Under some non-degeneracy condition, the Poisson brackets~\eqref{eqn:PBh} can be locally mapped to a canonical form where $\alpha$ is constant and $\beta=0$. This result can be viewed as the ``equivalent'' of the Lie-Darboux theorem for finite dimensional Hamiltonian systems. It is formalized using the following theorem~\cite{dubrovin_hamiltonian_1983}:  

\begin{theorem}(see Ref.~\cite{dubrovin_hamiltonian_1983})
    \label{thm:flat}
    For the bracket~\eqref{eqn:PBh} with non-degenerate $\alpha$ to satisfy the Jacobi identity, a necessary and sufficient condition is that there exists a locally invertible change of field variables ${\bf u}\mapsto {\bm \nu}={\bm \nu}({\bf u})$ such that the bracket expressed in the new coordinates ${\bm \nu}$ is of hydrodynamic type with a constant metric $g$ and no $\beta$-part, i.e., in the field variables ${\bm\nu}$, bracket~\eqref{eqn:PBh} becomes
    \begin{equation}
        \label{eqn:PBnu}
        \{\overline{F}, \overline{G}\} = \int  \partial_x \overline{F}_n g_{nm} \overline{G}_m\ {\rm d}x,
    \end{equation}
    where $\overline{F}[{\bm \nu}]=F[{\bf u}]$ and $\overline{F}_n={\delta \overline{F}}/{\delta \nu_n}$. 
\end{theorem}

\begin{lemme}
\label{lem:albe}
    The elements of the bracket~\eqref{eqn:PBh}, namely $\alpha_{nm}$ and $\beta_{nm}$ are given by
    \begin{eqnarray}
        && \alpha_{nm}({\bm\nu})=\frac{\partial u_n}{\partial\nu_k}g_{kl} \frac{\partial u_m}{\partial\nu_l},\\
        && \beta_{nm}({\bm\nu},\partial_x{\bm\nu})=\partial_x \frac{\partial u_n}{\partial\nu_k}g_{kl} \frac{\partial u_m}{\partial\nu_l},
    \end{eqnarray}
\end{lemme}
\begin{proof}
    This is a consequence of the identities~\eqref{eqn:change} since $\beta=0$ in the variables ${\bm\nu}$.
\end{proof}

A Casimir invariant is an observable that lies in the center of $\mathcal{A}$, i.e, $ \{C,F\} = 0$ for all $F\in \mathcal{A}$. As a consequence of Theorem~\ref{thm:flat}, a non-degenerate Poisson bracket of hydrodynamic type~\eqref{eqn:PBh} has as many independent Casimir invariants as the number of field variables, and they are of the form
$$
    C_n[{\bf u}]=\int \nu_n({\bf u}) \ {\rm d}x,
$$
for $n=1,\ldots,N$.

\begin{coro}
    Non-degenerate Poisson brackets of hydrodynamic type are classified by the signature $(l,q)$ of the metric $g$, where $q$ is the number of negative eigenvalues and $l$ is the number of positive ones. This signature is invariant by a change of coordinates of the type ${\bm \nu} \mapsto {\bm \eta}={\bm \eta}({\bm \nu})$.
\end{coro}

\begin{proof}
     By Sylvester's law of inertia, two congruent metrics have the same signature, and since under a coordinate change ${\bm \nu} \mapsto {\bm \eta}={\bm \eta}({\bm \nu})$, we have $\bar{g} = J^tgJ$ (see Eq.~\eqref{eqn:change_1}, where $J$ is the Jacobian of the change of variables, $\bar{g}$ and $g$ have the same signature.
\end{proof}

By extension, we define the signature of a Poisson bracket of hydrodynamic type~\eqref{eqn:PBh} as the signature of the associated metric $g$.

{\em Remark:} Using an appropriate rescaling and reshuffling of the field variables, a Poisson bracket of hydrodynamic type with signature $(l,q)$, can be cast into the form~\eqref{eqn:PBnu} with
$$
    g = \begin{pmatrix} {\mathbb I}_l&0\\ \:0&-{\mathbb I}_q\end{pmatrix},
$$
where ${\mathbb I}_l$ is the $l\times l$ identity matrix. 

\begin{defi}[Kinetic hydrodynamic bracket]
    A hydrodynamic bracket of the form~\eqref{eqn:PBh} comes from a kinetic model when there exists a local change of variables ${\bf u}\mapsto {\bf P}({\bf u})$ that maps it to a truncated Kupershmidt-Manin bracket~\eqref{eqn:PB_P} (i.e., truncated in the sense that the summations are restricted to $n,m=0,\ldots,N-1$).
\end{defi}

\begin{coro}
    \label{coro:PBksign}
    A non-degenerate hydrodynamic Poisson bracket~\eqref{eqn:PBh} coming from a kinetic model has a signature $(l,q)$ with $l,q\geq 1$.
\end{coro}

\begin{proof}
    By focusing on the subalgebra of observables $F[P_0,P_1]$ for the bracket~\eqref{eqn:PB_P}, the matrix $\alpha$ becomes
    $$
        \alpha = \begin{pmatrix}
                    0 & P_0 \\
                    P_0 & 2P_1
                \end{pmatrix},
    $$
    which has a negative determinant. Henceforth, it has a positive and a negative eigenvalue. 
\end{proof}
This corollary leads to the definition of partially decoupled brackets, which has the purpose of partially isolating $(P_0,P_1)$ from the rest.

\subsection{Partial decoupling of hydrodynamic Poisson brackets}

\begin{defi}[Partially decoupled bracket]
     A Poisson bracket of hydrodynamic type~\eqref{eqn:PBh} is said to be partially decoupled if it has the following form: 
     \begin{eqnarray}
         \{F,G\} &=& \int\left( \partial_xF_{u}G_{\rho} -F_{\rho} \partial_x G_{u} 
        + \partial_x w_k\left(F_u \frac{G_k}{\rho}-\frac{F_k}{\rho} G_u \right) \right. \nonumber\\
        && \qquad \left. + \partial_x \frac{F_k}{\rho} \alpha_{kl}({\bf w}) \frac{G_l}{\rho}  +  \frac{F_k}{\rho}  \beta_{kl}({\bf w}, \partial_x {\bf w}) \frac{G_l}{\rho}\right) \, {\rm d}x, \label{eqn:PBdec}
     \end{eqnarray}
    where the implicit sums over $k,l$ are from 1 to $N-2$ and $F_k = {\delta F}/{\delta w_k}$.
\end{defi}
Here the field variables are $(\rho, u, w_1,\ldots,w_{N-2})$. We notice that partially decoupled brackets are invariant under any change of field variables ${\bf w}\mapsto {\bf z}({\bf w})$. 

{\em Remark:} The set of functionals $F[\rho,u]$ and the set of functionals $F[\rho,{\bf w}]$ form subalgebras of the bracket~\eqref{eqn:PBdec}. This remark leads to the definition of the microscopic part of a partially decoupled hydrodynamic bracket with $N-2$ field variables:
\begin{equation}
    \label{eqn:PBmicro}
    \{F,G\}_{\rm m}=\int\left( \partial_x F_k \alpha_{kl}({\bf w}) G_l  +  F_k  \beta_{kl}({\bf w}, \partial_x {\bf w}) G_l\right) \, {\rm d}x,
\end{equation}
i.e., where the macroscopic degree of freedom represented by the fluid density and fluid velocity $(\rho, u)$ has been eliminated. 

\begin{proposition}
    The bracket~\eqref{eqn:PBdec} with $N$ field variables is a Poisson bracket if and only if the microscopic bracket~\eqref{eqn:PBmicro} with $N-2$ field variables is a Poisson bracket.
\end{proposition} 
\begin{proof}
    Suppose first that the bracket~\eqref{eqn:PBdec} is Poisson. Since the set of functionals $F[\rho,{\bf w}]$ is a subalgebra, it follows that the induced bracket~\eqref{eqn:PBmicro} is a Poisson bracket.\\
    Conversely, suppose that bracket~\eqref{eqn:PBmicro} is a Poisson bracket. By Theorem~\ref{thm:flat} there exists a local change of variables ${\bf w}\mapsto {\bm \nu}({\bf w})$ such that $\{F,G\}_{\rm m} = \int g_{kl} \partial_x F_k\ G_l \, {\rm d}x$, where $F_k=\delta F / \delta \nu_k$.\\
    In the variables $(\rho,\psi,{\bm \nu})$, the bracket $\{F,G\}= \int (\partial_x F_{\psi} G_{\rho} - F_{\rho} \partial_x G_{\psi} + \partial_x F_k g_{kl} \ G_l) \ {\rm d}x$ is of Poisson type since the coefficients do not depend on the field variables. Under the following changes of variables $(\rho,\psi,{\bm \nu}) \mapsto (\rho,u,{\bm \eta})$ where $u={\bm \nu}\cdot g^{-1} {\bm \nu}/(2\rho)+ \psi$ and $\eta_k=\nu_k / \rho$ and then ${\bm \nu}\mapsto {\bf w}({\bm \nu})$, this bracket is mapped to bracket~\eqref{eqn:PBh}.
\end{proof}

\begin{proposition}
    Any hydrodynamic bracket~\eqref{eqn:PBh} with signature $(l,q)$ where $l,q\geq1$ can be cast into a partially decoupled hydrodynamic bracket. In particular, any hydrodynamic bracket coming from a kinetic equation, such as Vlasov, can be cast into a partially decoupled hydrodynamic bracket.
\end{proposition}

\begin{proof}
    We consider a hydrodynamic Poisson bracket~\eqref{eqn:PBh} with signature $(l,q)$ with $l,q\geq1$. Using Corollary~\ref{coro:PBksign}, it is rewritten in the ${\bm \nu}$ variables as 
    $$
        \{F,G\} = \int \left( \partial_x F_1 G_1 - \partial_x F_2 G_2 + \varepsilon_k\partial_x F_k G_k \right) \ {\rm d}x,
    $$
    where $\epsilon_k=\pm 1$ for $k=3,\ldots,N$.
    Under the following change of variables, $\rho = (\nu_1+\nu_2)/\sqrt{2}$, $u = (\nu_1-\nu_2)/\sqrt{2} + (\sum_{k\geq3}\varepsilon_k \nu_k^2)/ (2\rho)$ and $\bar{\nu}_k = {\nu_k}/{\rho}$, it takes the partially decoupled form~\eqref{eqn:PBdec}.
\end{proof}
It follows that a non-degenerate hydrodynamic bracket of kinetic origin can be cast in the following form
\begin{equation}
     \{F,G\} = \int\left( \partial_xF_{u}G_{\rho} -F_{\rho} \partial_x G_{u} 
    + \partial_x \nu_k\left(F_u \frac{G_k}{\rho}-\frac{F_k}{\rho} G_u \right) + \partial_x \frac{F_k}{\rho} g_{kl} \frac{G_l}{\rho}\right) \, {\rm d}x, \label{eqn:PBdec_nu}
 \end{equation}
 where $F_k={\delta F}/{\delta \nu_k}$ and $g$ is a constant symmetric matrix.

\begin{lemme}
\label{lem:cas}
    A non-degenerate partially decoupled Poisson bracket of hydrodynamic type~\eqref{eqn:PBh} has $N-2$ Casimir invariants of the form 
    $$
        C[\rho, {\bm \nu}] = \int \rho {\bm \nu} \ {\rm d}x,
    $$
    in addition to the two Casimir invariants
    \begin{align}
        & C[\rho] = \int \rho \ {\rm d}x,\\
        & C[\rho,u,{\bm \nu}] = \int \left( u - \frac{\rho}{2}{\bm \nu}\cdot g^{-1}{\bm \nu}\right) \ {\rm d}x, \label{eqn:CasPsi}
    \end{align}
    where ${\bm \nu} = {\bm \nu}({\bf w})$ are the coordinates which flatten the metric of bracket~\eqref{eqn:PBmicro} into the metric $g$ as in Eq.~\eqref{eqn:PBdec_nu}.
\end{lemme}

\begin{proof}
    A Casimir invariant $C = \int f({\bf w})\, {\rm d}x$ of the bracket~\eqref{eqn:PBmicro} yields a Casimir invariant $C = \int \rho f({\bf w})\,{\rm d}x$ of the bracket~\eqref{eqn:PBdec}.
    The conservation of the Casimir invariant $C[\rho,u,{\bm \nu}]$ can be verified by a direct and straightforward calculation of $\{C,\rho\}$, $\{C,u\}$ and $\{C,{\bm \nu}\}$. It then follows from the Leibniz rule that $\{C,F\}=0$ for all $F$.
\end{proof}

\subsection{Centered moments}

To write a hydrodynamic bracket coming from a kinetic model into a partially decoupled hydrodynamic bracket, a strategy used in Refs.~\cite{perin_higher-order_2014-1,perin_hamiltonian_2015-2,perin_hamiltonian_2015-1,chandre_four-field_2022} is to move to the moments centered around the velocity, i.e.,
\begin{equation}
\label{eqn:S}
    S_n(t,x) = \frac{1}{\rho^{n+1}}\int (v-u)^n f(t,x,v)\ {\rm d}v,
\end{equation}
where the fluid velocity is given by
$$
    u = \frac{1}{\rho}\int v f(t,x,v)\ {\rm d}v.
$$
In this way, we notice that $S_0=1$ and $S_1=0$. The centered moments ${\bf S}$ are obtained from the moments ${\bf P}$ as
$$
    S_n = \frac{1}{\rho^{n+1}}\sum_{k=0}^n \binom{n}{k} (-u)^{n-k} P_k ,
$$
and inverted as
$$
    P_n = \sum_{k=0}^n \binom{n}{k} \rho^{k+1} u^{n-k} S_k .
$$
In particular, we notice that $S_n$ depends on $(P_0,\ldots,P_n)$, and vice versa, $P_n$ depends on $(\rho, u, S_2,\ldots,S_n)$. 
Using the change of variables $(P_0,P_1,\ldots,P_{N-1})\mapsto (\rho, u, S_2, \ldots, S_{N-1})$, the bracket becomes a partially decoupled hydrodynamic bracket of the form~\eqref{eqn:PBdec} with 
\begin{eqnarray*}
    && \alpha_{nm}({\bf S})= (n+m) S_{n+m-1} -m(n+1)S_n S_{m-1}-n(m+1)S_m S_{n-1},\\
    && \beta_{nm}({\bf S},\partial_x {\bf S}) = n \partial_x S_{n+m-1} -n S_{n-1} \partial_x S_{m} -n(m+1) S_{m} \partial_x S_{n-1}-nm S_{m-1} \partial_x S_{n},
\end{eqnarray*}
where the closure is simplified to
$$
    S_{n\geq N}=S_{n\geq N}(S_2,\ldots, S_{N-1}).
$$
In particular, we notice that these closures do not depend on $\rho$ and $u$. 

Following Lemma~\ref{lem:cas}, another interesting set of centered variables is obtained by centering the moments not with respect to $u$ but with respect to $\psi(x) \equiv u-\rho \mu_1$ such that $C=\int \psi(x) \ {\rm d}x$ given by Eq.~\eqref{eqn:CasPsi} is a Casimir invariant of the bracket. Imposing that $C$ is a Casimir invariant defines the variable $\mu_1$. For the bracket expressed in the ${\bm\nu}$-variables, it takes the form
\begin{equation}
    \label{eqn:mu1}
    \mu_1(x)=\frac{1}{2}{\bm\nu} \cdot g^{-1} {\bm\nu},
\end{equation}
which is a homogeneous quadratic polynomial in the $N-2$ variables $\nu_k$. 
We also define $\mu_n$ for $n\geq 2$ as 
\begin{equation}
\label{eqn:mu}
    \mu_n(t,x) = \frac{1}{\rho^{n+1}}\int (v-\psi)^n f(t,x,v) \ {\rm d}v.
\end{equation}
These variables are linked to the ones used in Ref.~\cite{burby_variable-moment_2023}.  
The ${\bm \nu}$-variables are linked to the ${\bf S}$ variables in the following way:
$$
    S_n=\sum_{k=0}^n \binom{n}{k}\left(-\mu_1\right)^{n-k}\mu_k.
$$
Next we perform the change of variables $(\rho,u,S_2,\ldots,S_{N-1})\mapsto (\rho, u, \mu_1,\ldots,\mu_{N-2})$. In Appendix~\ref{app:PBmu}, we show that the bracket expressed in these variables remains of the type~\eqref{eqn:PBdec} with
\begin{subequations}
\label{eqn:albemu}
    \begin{eqnarray}
    && \alpha_{nm}({\bm\mu}) = (n+m)\mu_{n+m-1}-m\mu_{m-1}\gamma_n-n\mu_{n-1}\gamma_m,\\
    && \beta_{nm}({\bm\mu},\partial_x {\bm\mu}) = n\partial_x\mu_{n+m-1}-n\partial_x\mu_{n-1}\gamma_m-m\mu_{m-1}\partial_x\gamma_n,
\end{eqnarray}
\end{subequations}
where the explicit expression of $\gamma_n$ depends on the specific closure. In the variables $\nu_k$, it becomes 
\begin{equation}
\label{eqn:gamma}
    \gamma_n = (n+1)\mu_n -\nu_k\frac{\partial \mu_n}{\partial \nu_k}.
\end{equation}
If the parametric expression of $\mu_n$ is a homogeneous polynomial of degree $n+1$, we have that $\gamma_n=0$. As a consequence, the equations for the elements $\alpha$ and $\beta$ in $\bm\mu$ variables are identical as the equations in the $\bf P$ variables (see Eq.~\eqref{eqn:albeP}). This is the case for some of the cases examined below, such as the multi-delta closure and Burby's closure~\cite{burby_variable-moment_2023}. However, there exist cases where $\gamma_n\not= 0$, such as the waterbags. 

In what follows, we use the notation for the field variables and the closure functions as follows:
\begin{itemize}
    \item $P_n$: moments of $f$ in velocity given by Eq.~\eqref{eqn:moments},
    \item $S_n$: centered moments of $f$ around the fluid velocity $u$, given by Eq.~\eqref{eqn:S},
    \item $\mu_n$: centered moments of $f$ around $\psi=u-\rho \mu_1$, given by Eq.~\eqref{eqn:mu},
    \item Greek letters such as $\nu_k$, $\eta_k$ or $\xi_k$ for the field variables which flatten $\alpha$ into a constant metric $g$ (and consequently no $\beta$-part).
\end{itemize}
Also, in the rest of the article, $N$ always designates the number of field variables of the main fluid model (i.e., including $\rho$ and fluid velocity).

\begin{proposition}
     The closure, i.e., $\mu_n=\mu_n(\nu_1,\ldots,\nu_{N-2})$ for $n\geq 1$, is generated by a single moment $\mu_2$. 
\end{proposition}

\begin{proof}
    We assume that the hydrodynamic bracket~\eqref{eqn:PBdec} with elements $\alpha$ and $\beta$ given by Eq.~\eqref{eqn:albemu} is a Poisson bracket. The change of variables from $\bm\mu$ to $\bm\nu$ ensures that 
    \begin{eqnarray*}
        && \frac{\partial \mu_n}{\partial \nu_k} g_{kl} \frac{\partial \mu_m}{\partial \nu_l} = (n+m)\mu_{n+m-1}-m\mu_{m-1}\gamma_n-n\mu_{n-1}\gamma_m,\\
        && \partial_x \frac{\partial \mu_n}{\partial \nu_k} g_{kl} \frac{\partial \mu_m}{\partial \nu_l} = n\partial_x\mu_{n+m-1}-n\partial_x\mu_{n-1}\gamma_m-m\mu_{m-1}\partial_x\gamma_n.
    \end{eqnarray*}
    For  $m=2$, the first identity becomes
    $$
        \mu_{n+1}=\frac{1}{n+2} \left[\frac{\partial \mu_n}{\partial \nu_k} g_{kl} \frac{\partial \mu_2}{\partial \nu_l} +2\mu_{1}\gamma_n+n\mu_{n-1}\gamma_2 \right],
    $$
    where $\gamma_n$ is given by Eq.~\eqref{eqn:gamma} and $\mu_1$ given by Eq.~\eqref{eqn:mu1}. Since $\gamma_n$ only depends on $\mu_n$ and its first derivatives, this recurrence relation ensures that all $\mu_{n\geq 3}$ are uniquely determined by $\mu_2$. 
\end{proof}

\section{Classification of known Hamiltonian fluid closures}
\label{sec:part2}

In this section, we investigate all the known Hamiltonian closures in the light of the results listed in Sec.~\ref{sec:part1}. 

\subsection{The multi-delta distribution}
\label{sec:part1delta}

We consider the Vlasov distribution (as considered in Refs.~\cite{Gosse2003,fox_higher-order_2009,yuan_conditional_2011,chalons_beyond_2012,cheng_class_2014,inglebert_multi-stream_2011, inglebert_multi-stream_2012, inglebert_electron_2012,antoine_multi-stream_2025}) given by $N=2M$ fields $a_k(x)$ and $v_k(x)$ for $k=1,\ldots, M$:
\begin{equation}
\label{eqn:deltas}
    f(t,x,v) = \sum_{k=1}^M a_k(t,x)\delta(v-v_k(t,x)),
\end{equation}
where $M$ is the number of Dirac delta functions. From the Poisson bracket~\eqref{eqn:PB_Vlasov}, the Poisson bracket associated with the distribution~\eqref{eqn:deltas} is given by:
\begin{equation}
\label{eqn:PBdelta}
    \{F,G\} =\int \left( \partial_x F_{a_k}G_{v_k} - \partial_x G_{a_k} F_{v_k}\right)\,{\rm d}x. 
\end{equation}
We notice that the above bracket is readily of hydrodynamic form~\eqref{eqn:PBh} and its metric is flat, i.e., $\beta=0$ and $\alpha$ is given by
$$
    \alpha = \begin{pmatrix}0&{\mathbb I}_{M}\\ {\mathbb I}_{M} &0\end{pmatrix}.
$$
Consequently, the signature of the bracket~\eqref{eqn:PBdelta} is $(M,M)$. 
The moments ${\bf P}$ are given by
$$
    P_n = \sum_{k=1}^M a_k v_k^n.
$$
We notice that the moment $P_n$ is a homogeneous polynomial of degree $n+1$ in $N=2M$ variables. 
Next, we proceed with the partial decoupling by considering the fluid density and fluid velocity. We consider the following variables:
\begin{align*}
    \rho &= \sum_{l=1}^M a_l,\\ 
    u &= \frac{1}{\rho}\sum_{l=1}^{M}a_l v_l,\\
    \xi_k &= \frac{a_k}{\rho},\\
    \eta_k &= \frac{v_k-v_1}{\rho},
\end{align*}
for $k=2,\ldots,M$. This change of variables is invertible, and its inverse is given by
\begin{align*}
    a_1 &= \rho\left(1-\sum_{l=2}^M \xi_l\right),\\ 
    v_1 &= u- \rho\sum_{l=2}^{M}\xi_l \eta_l,\\
    a_k &= \rho \xi_k,\\
    v_k &= u- \rho\sum_{l=2}^{M}\xi_l \eta_l +\rho \eta_k,
\end{align*}
for $k=2,\ldots, N-1$.
In the normal variables $\boldsymbol{\nu} =(\boldsymbol{\xi},\boldsymbol{\eta})$, the bracket becomes of the form~\eqref{eqn:PBdec_nu} with 
$$
    g = \begin{pmatrix}0&{\mathbb I}_{M-1}\\ {\mathbb I}_{M-1} &0\end{pmatrix}.
$$
As a consequence, the signature of the microscopic bracket~\eqref{eqn:PBmicro} is $(M-1,M-1)$.

The centered fluid moments ${\bf S}$ are expressed as, 
\begin{align*}
    S_n &= \frac{1}{\rho^{n+1}}\sum_{i=1}^M a_i(v_i-u)^n,\\
    &= \left(1-\sum_{k=2}^M \xi_k\right)\left(-\sum_{k=2}^M \eta_k \xi_k\right)^n + \sum_{k=2}^M \xi_k\left(\eta_k - \sum_{l=2}^M \xi_l \eta_l\right)^n.
\end{align*}

{\em Remark:} $S_n$ is a polynomial of degree $2n$ in the $2M-2$ normal variables $({\bm \xi},{\bm\eta})$.

Next we compute the $\bm\mu$-variables by centering the moments around $\psi=u-\rho\mu_1$ where $\mu_1$ is given by Eq.~\eqref{eqn:mu1} as $\mu_1 \equiv \sum_{k=2}^M \xi_k \eta_k$, i.e.,
\begin{align*}
    \mu_n &= \frac{1}{\rho^{n+1}}\sum_{k=1}^M a_k(v_k-u+\rho\mu_1)^n,\\
    &= \sum_{k=2}^M \xi_k \eta_k^{n},
\end{align*}
since $v_1-u+\rho\mu_1=0$. Given that $\mu_n$ is a homogeneous polynomial of degree $n+1$, the functions $\gamma_n$ given by Eq.~\eqref{eqn:gamma} vanish. Consequently, under the changes of variables $(\rho,u,\xi_2,\ldots,\xi_{M},\eta_2,\ldots,\eta_M) \to (\rho,u,\mu_1,\ldots,\mu_{N-2})$, one gets the partially decoupled hydrodynamic bracket of the form~\eqref{eqn:PBdec} with
\begin{eqnarray*}
    && \alpha_{nm}({\bm \mu}) = (n+m)\mu_{n+m-1},\\
    && \beta_{nm}({\bm \mu},\partial_x {\bm \mu}) = m\partial_x \mu_{n+m-1},
\end{eqnarray*}
for $n,m=1,\ldots,N-2$ (see also Appendix~\ref{app:PBmu}).

We notice that the expression of the moments $\mu_n$ in the $2M-2$ variables $(\bm \xi,\bm\eta)$ is the same as the moments $P_n$ in the $2M$ variables $(\bf a, \bf v)$. 

The cubic polynomial generating the closure is
$$
    \mu_2 = \sum_{k=2}^M \xi_k \eta_k^{2}. 
$$

\subsection{The waterbag distribution}

We consider the Vlasov distributions given by the waterbag ansatz (first introduced in Refs.~\cite{depackh_water-bag_1962, bertrand_non_1968, bertrand_frequency_1969, bertrand_non-linear_1976}), i.e., a piecewise constant function in the velocity domain,
$$
    f(t,x,v) = \sum_{n=1}^{N} a_n\Theta\left(v-v_n(t,x)\right),
$$
where $a_n$ are constant, $\Theta$ is the Heaviside distribution, and $v_n(t,x)$ are the contour velocities of the $N-1$ bags of height $\sigma_k=\sum_{n=1}^{k} a_n$ for $k=1,\ldots, N-1$. To ensure a compact support, we impose the condition $\sigma_N=\sum_{n=1}^{N} a_n =0$.

In the algebra of functionals of ${\bf v} = (v_1,\cdots,v_{N})$, the Poisson bracket is given by~\cite{yu_waterbag_2000,chesnokov_reductions_2012,morrison2013-1,morrison2013-2}
$$
    \{F,G\}_{{\rm wb}} = - \int \sum_{n=1}^{N}\frac{1}{a_n}\partial_x F_n G_n\ {\rm d}x,
$$
where $F_n={\delta F}/{\delta v_n}$. We notice that this bracket is a hydrodynamic bracket~\eqref{eqn:PBh} with $\beta=0$ and $\alpha_{nm}=-\delta_{nm}/a_n$. As a consequence, $N$ independent Casimir invariants in the variables $v_n$ are given by
\begin{equation}
\label{eqn:CasWB}
    C_n = \int v_n \ {\rm d}x.
\end{equation}
We observe that the moments~\eqref{eqn:moments}, given by
$$ 
    P_n = -\frac{1}{n+1}\sum_{n=1}^{N} a_n v_n^{n+1},
$$
are homogeneous polynomials of degree $n+1$ in the $N$ variables $v_n$.

Next, we cast this Poisson bracket into a partially decoupled form~\eqref{eqn:PBdec} by introducing the normal variables :
\begin{eqnarray*}
    && \rho = -\sum_{n=1}^{N} a_n v_n,\\
    && u = -\frac{1}{2\rho} \sum_{n=1}^{N} a_n v_n^2,\\
    && \nu_k = \frac{1}{\rho}\sum_{l=1}^{k} \sigma_l (v_{l+1}-v_l),
\end{eqnarray*}
for $k=1,\ldots,N-2$.
We notice that $\nu_k$ is the cumulative normalized particle density up to the $k$-th bag and, using the definition of $\nu_k$, we have $\nu_{N-1}=1$ and $\nu_0=0$. This change of variables is invertible and its inverse is given by
\begin{subequations}
\label{eqn:WBinv}
    \begin{eqnarray}
    && v_1 = u+\frac{\rho}{2}\sum_{k=2}^N a_k\left(\sum_{l=1}^{k-1}\frac{\nu_l-\nu_{l-1}}{\sigma_l} \right)^2,\\
    && v_k = v_1 + \rho \sum_{l=1}^{k-1}\frac{\nu_l-\nu_{l-1}}{\sigma_l},
\end{eqnarray}
\end{subequations}
for $k=2,\ldots,N$.
By performing the change of variables $(v_1,\ldots,v_N)\mapsto (\rho, u, \nu_1,\ldots, \nu_{N-2})$, the Poisson bracket becomes a partially decoupled hydrodynamic bracket~\eqref{eqn:PBdec_nu} with
$$
    g_{nm} =-\delta_{nm}  \frac{\sigma_n \sigma_{n+1}}{a_{n+1}}.
$$
Since $\sigma_n$ corresponds to the height of the $n$-th waterbag, the sign of $g_{nm}$ is the opposite of the sign of $a_{n+1}$. Consequently, the signature of $\{\cdot ,\cdot \}_{{\rm wb}}$ is $(N-k,k)$ where $k\in \{1,\cdots,N-1\}$ is the number of positive $a$'s. This means that the signature of the bracket is directly linked to the predefined number of increasing vs decreasing bags in the distribution. Since the heights of the bags are constant, this signature is a characteristic of the shape of the chosen initial distribution. 

The centered fluids moments ${\bf S}$ are given by~\cite{perin_hamiltonian_2015-1}
\begin{align*}
    S_n &= -\frac{1}{(n+1)\rho^{n+1}}\sum_{k=1}^{N} a_k (v_k-u)^{n+1},\\
        &= -\frac{1}{(n+1)}\sum_{k=1}^{N} a_k \eta_k^{n+1}(\nu_1,\cdots,\nu_{N-2}),
\end{align*}
where, using Eq.~\eqref{eqn:WBinv}, the expressions of $\eta_k$ are given by
\begin{align*}
    \eta_{1} &= \frac{1}{2}\sum_{k=2}^{N}a_k\left(\sum_{l=1}^{k-1}\frac{\nu_{l}-\nu_{l-1}}{\sigma_l}\right)^2,\\
    \eta_{k} &=\eta_{1} +\sum_{l=1}^{k-1}\frac{\nu_{l}-\nu_{l-1}}{\sigma_l},
\end{align*}
for $k=2,\ldots,N$, where $\nu_0=0$ and $\nu_{N-1}=1$. 

{\em Remark:} $S_n$ is a polynomial of degree $2n$ in the normal variables ${\bm \nu}$. We notice that for this closure, the polynomial is not homogeneous since 
$$
    S_n({\bm\nu}=0)=\frac{1+(-1)^n}{(n+1)2^{n+1}a_N^n}.
$$
Instead of centering the moments around $u$, we center the moments around $\psi$ defined by $\psi=u-\rho\mu_1$. The first moment $\mu_1$ is given by Eq.~\eqref{eqn:mu1}:
$$
    \mu_1=\frac{1}{2}\sum_{k=1}^{N-2}\frac{\nu_k^2}{\lambda_k},
$$
where $\lambda_k=-\sigma_k \sigma_{k+1}/a_{k+1}$.
The explicit expression of $\psi$ as a function of $v_n$ becomes
$$
    \psi = \frac{v_N}{2}+\frac{\rho}{2 a_N}=-\frac{1}{2a_N}\sum_{n=1}^{N-1} a_n v_n,
$$
where, as expected, it can be shown to be a Casimir invariant given Eq.~\eqref{eqn:CasWB}.
The moments $\mu_n$ for $n\geq 2$ as functions of the variables $\nu_k$ are given by
\begin{eqnarray*}
    \mu_n &=& -\frac{1}{n+1}\frac{1}{\rho^{n+1}}\sum_{k=1}^N a_k \left(v_k-\psi \right)^{n+1},\\
    &=& \frac{(-1)^n}{n+1}\left(\sum_{k=1}^{N-1}a_k \left(\frac{1}{2a_N}+\sum_{l=k}^{N-1}\frac{\nu_l-\nu_{l-1}}{\sigma_l} \right)^{n+1} + \frac{1}{2^{n+1}a_N^n} \right).
\end{eqnarray*}
From this expression, we notice that $\mu_n$ is a polynomial of degree $n+1$ in the $N-2$ variables $\nu_k$.
In the variables ${\bm \mu}$ the brackets takes the form~\eqref{eqn:PBdec} with
\begin{eqnarray*}
    && \alpha_{nm}=(n+m)\mu_{n+m-1}-m\mu_{m-1}\Lambda^n -n\mu_{n-1}\Lambda^m+2\Lambda nm \mu_{n-1}\mu_{m-1},\\
    && \beta_{nm}=n\partial_x \mu_{n+m-1} -n\partial_x \mu_{n-1}\Lambda^m+2\Lambda nm \partial_x \mu_{n-1}\mu_{m-1},
\end{eqnarray*}
where $\Lambda=-1/(2a_N)$. The inhomogeneous character of the polynomial $\mu_n$ is expressed by a non-zero $\gamma_n$ where $\gamma_n$ is given by Eq.~\eqref{eqn:gamma}
$$
    \gamma_n=\Lambda^n-n\Lambda \mu_{n-1},
$$
whereas this quantity is zero for all other known closures.

Although the coefficients $\alpha$ and $\beta$ depend only on the constant $\Lambda$, the realization of the moment variables $\boldsymbol{\mu}$ as functions of the normal variables $\boldsymbol{\nu}$ depends explicitly on the signs of the waterbag heights. Consequently, two waterbag distributions with identical $\Lambda$ but different signatures give rise to non-isomorphic Poisson structures, despite having formally identical reduced brackets in the $\boldsymbol{\mu}$-variables.

The cubic polynomial generating the closure is 
$$
    \mu_2 = \frac{1}{3}\left(\sum_{k=1}^{N-1}a_k \left(\frac{1}{2a_N}+\sum_{l=k}^{N-1}\frac{\nu_l-\nu_{l-1}}{\sigma_l} \right)^{3} + \frac{1}{8 a_N^2} \right).
$$

\subsection{Burby's Hamiltonian closure (see Ref.~\texorpdfstring{\cite{burby_variable-moment_2023}}{[Burby, 2023]})}
\label{sec:part1JB}

In this section we consider the Hamiltonian closure obtained in Ref.~\cite{burby_variable-moment_2023}. In the one-dimensional case, this closure corresponds to considering a Vlasov distribution of the form 
$$
    f(t,x,v) = \sum_{k=0}^{N-2}\frac{(-1)^k}{k!}\mu_k(t,x) \delta^{(k)}(v-\psi(t,x)),
$$
where $\delta^{(k)}$ is the $k$-th derivative of the Dirac delta function defined as
$$
    \int \delta^{(k)}(v) \phi(v) \ {\rm d}v = (-1)^k \phi^{(k)}(0),
$$
for an arbitrary function $\phi \in C^\infty ({\mathbb R})$. The moments $P_n$ are given by
$$
    P_n=\sum_{k=0}^n \binom{n}{k} \mu_k\psi^{n-k}.
$$
The centered moments $S_n$ are given by
$$
    S_n = \sum_{k=0}^n \binom{n}{k} \overline{\mu}_k \left(-\overline{\mu}_1 \right)^{n-k},
$$
where $\overline{\mu}_k=\mu_k/\rho^{k+1}$. As expected, we have $S_0=\overline{\mu}_0=1$ and $S_1=0$. Therefore the centered moments $S_n$ are only functions of $\overline{\mu}_k$ for $k=1,\ldots,N-2$. 

{\em Remark:} If $\overline{\mu}_n$ is a polynomial of degree $n+1$ for $n\geq 1$, then $S_n$ is a polynomial of degree $2n$. Therefore we notice the advantage of working with the variables $\overline{\mu}_n$ rather than the centered variables $S_n$, to obtain closure functions of the lowest possible degree.

We note that the first fluid moments are given by
\begin{eqnarray*}
    && \rho = \mu_0,\\
    && u = \frac{\mu_1}{\rho} + \psi,
\end{eqnarray*}
In the variables $(\rho, \psi, \mu_1,\ldots, \mu_{N-2})$, the Poisson bracket becomes
\begin{equation}
\label{eqn:PBJB}
    \{F,G\} = \int \left( \partial_x F_{\psi}G_{\rho} - F_{\rho}\partial_x G_{\psi}+ \partial_x F_k \alpha_{kl}({\bm\mu}) G_l + F_k \beta_{kl}({\bm\mu}, \partial_x{\bm\mu}) G_l \right) \ {\rm d}x,
\end{equation}
where the implicit summations over $k$ and $l$ run from 1 to $N-2$, and
\begin{subequations}
\label{eqn:JBalbe}
    \begin{eqnarray}
    && \alpha_{kl}({\bm\mu}) = (k+l)\mu_{k+l-1},\\
    && \beta_{kl}({\bm\mu}, \partial_x{\bm\mu}) = k\partial_x \mu_{k+l-1},
\end{eqnarray}
\end{subequations}
and $\mu_k=0$ for $k\geq N-1$. In other terms, the bracket is of hydrodynamic type and the elements $\alpha$ and $\beta$ are upper anti-triangular. Contrary to the examples of the multi-delta distribution and the waterbags, the expression of the Poisson bracket has an explicit dependence on the field variables, so computing the signature of the bracket in the original variables is not trivial.

In order to proceed, we consider the following change of variables: $(\rho, \psi, \mu_1,\cdots,\mu_{N-2}) \mapsto (\rho,u, \bar{\mu}_1 , \cdots, \bar{\mu}_{N-2})$ defined by
\begin{eqnarray*}
    && u = \frac{\mu_1}{\rho} + \psi, \\
    && \bar{\mu}_k = \frac{\mu_k}{\rho^{k+1}}, 
\end{eqnarray*}
for $k=1,\ldots, N-2$. This leads to the partial decoupling of the hydrodynamic bracket in the form given by Eq.~\eqref{eqn:PBdec}, where the expressions of the elements $\bar\alpha$ and $\bar\beta$ satisfy Eq.~\eqref{eqn:JBalbe} with $\bar{\bm\mu}$ instead of ${\bm\mu}$. In this way, the Poisson bracket is already in the reduced variables $\bar{\bm\mu}$. 

\subsubsection{Expression of the Casimir invariants}

Since the bracket~\eqref{eqn:PBJB} satisfies the Jacobi identity, it has $N$ Casimir invariants. Two obvious Casimir invariants are
\begin{eqnarray*}
    && C[\rho]=\int \rho \ {\rm d}x,\\
    && C[\psi]= \int \psi \ {\rm d}x.
\end{eqnarray*}
The other $N-2$ Casimir invariants related to the specific shape of the elements $\alpha$ and $\beta$ are less trivial, but are of the form 
$$
    C[\rho,\bar{\bm\mu}]=\int \rho {\bm\nu}(\bar{\bm\mu})\  {\rm d}x,
$$
as mentioned in Lemma~\ref{lem:cas}.
The goal of this section to obtain their explicit expressions as functions of $\bar{\bm\mu}$. In what follows, we drop the bars above ${\bm\mu}$, and the scaling with $\rho$ is implicitly assumed. 
 
Finding the expressions of the Casimir invariants amounts to finding the change of variables ${\bm\mu}\mapsto {\bm\nu}={\bm\nu}({\bm\mu})$, or alternatively to determine the parametrization ${\bm\mu} = {\bm\mu}({\bm\nu})$ which maps the bracket~\eqref{eqn:PBdec} into the bracket~\eqref{eqn:PBdec_nu} and invert it. 
The equations to be solved are obtained from Lemma~\ref{lem:albe} and Eq.~\eqref{eqn:JBalbe} for $m\equiv N-2$:
\begin{subequations}
\label{eqn:JBeqn}
    \begin{eqnarray}
    && \frac{\partial \mu_k}{\partial\nu_i}g_{ij}\frac{\partial \mu_l}{\partial\nu_{j}} = (k+l)\mu_{k+l-1},\label{eqn:JBeqn_a} \\
    && \partial_x \frac{\partial \mu_k}{\partial\nu_i}g_{ij}\frac{\partial \mu_l}{\partial\nu_{j}} = k\partial_x \mu_{k+l-1}, \label{eqn:JBeqn_b}
\end{eqnarray}
\end{subequations}
for $k,l=1,\ldots,m$, and where $\mu_k=0$ for $k\geq m+1$. The implicit summations over $i$ and $j$ are taken from 1 to $m$.
First, we assume the matrix $g$ to be the following $m\times m$ antidiagonal matrix:
\begin{equation}
\label{eqn:JBg}
    g =
        \begin{pmatrix}
        0      & \cdots & 0      & 1      \\
        \vdots & \iddots & 1      & 0      \\
        0      & 1      & \iddots & \vdots \\
        1      & 0      & \cdots & 0
        \end{pmatrix}.
\end{equation}
The identities~\eqref{eqn:JBeqn} to be solved simplify to:
\begin{eqnarray*}
    && \frac{\partial \mu_k}{\partial\nu_j}\frac{\partial \mu_l}{\partial\nu_{m+1-j}} = (k+l)\mu_{k+l-1}, \\
    && \partial_x \frac{\partial \mu_k}{\partial\nu_j}\frac{\partial \mu_l}{\partial\nu_{m+1-j}} = k\partial_x \mu_{k+l-1}.
\end{eqnarray*}
The goal is to find the solution of these coupled partial differential equations for an arbitrary $m$.

\begin{lemme}
\label{prop:munu}
    The functions $\mu_n^{(m)}$ of the $m+1-n$ last variables ${\bm\nu}$, i.e.,
    $$
        \mu_n^{(m)}=\mu_n^{(m)}(\nu_n,\ldots,\nu_m),
    $$
    for $n=1,\ldots, m$, satisfy the following relations
    \begin{eqnarray*}
        && \frac{\partial \mu_k^{(m)}}{\partial\nu_j}\frac{\partial \mu_l^{(m)}}{\partial\nu_{m+1-j}} = 0,\\
        && \partial_x \frac{\partial \mu_k^{(m)}}{\partial\nu_j}\frac{\partial \mu_l^{(m)}}{\partial\nu_{m+1-j}} =0,
    \end{eqnarray*}
    for $k+l > m+1$.
\end{lemme}
\begin{proof}
    The derivative ${\partial \mu_k^{(m)}}/{\partial\nu_j}$ is nonzero if $j\geq k$, and ${\partial \mu_l^{(m)}}/{\partial\nu_{m+1-j}}$ is nonzero if $m+1-j\geq l$. The condition $m+1-j\geq l$ together with $k+l > m+1$ leads to $j< k$. So there is no possible index $j$ such that both partial derivatives are non-zero. 
\end{proof}
We look for solutions to Eq.~\eqref{eqn:JBeqn} of this form. In what follows, we consider functions $\mu_n$ for different values of $m$, i.e., we define the family $\cal M$ of centered moments $\mu_n^{(m)}$ for $n=1,\ldots,m$ and $m=1,\ldots, N-2$. Here $m$ is referred to as the level of the moment(s), and therefore $\mu_n^{(m)}$ is the $n$-th centered moment of level $m$.

\begin{lemme}
\label{prop:murel}
    Let $\cal M$ be a family of centered moments $\mu_n^{(m)}(\bm\nu)$ for $n=1,\ldots, m$ and $m=1,\ldots,M$, such that $\mu_n^{(m)}(\bf 0) = 0$. The two proposals are equivalent: \\
    {\rm ($i$)} $\cal M$ satisfies the following differential identities:
    \begin{equation}
    \label{eqn:dmu_JB}
        \frac{\partial}{\partial \nu_k}\left[\mu_n^{(m)}(\nu_n, \cdots, \nu_m)\right] = n\mu_{n-1}^{(k-1)}(\nu_{m-k+n}, \cdots, \nu_{m}),
    \end{equation}
    for all $k=n,\ldots,m$, with $\mu_0^{(l)}(x_0,\ldots,x_l) = x_0$ for all levels $l$.\\
    {\rm ($ii$)} $\cal M$ satisfies the following recursive relations:
    $$
        \mu_n^{(m)}(\nu_n, \ldots, \nu_m) =
        \begin{cases}
        \displaystyle\frac{\nu_m^{m+1}}{m+1} & \text{if } n = m, \\
        \displaystyle\sum_{k=0}^{n} \binom{n}{k} \nu_m^{n-k} \mu_k^{(m - n - 1)}(\nu_{k+n}, \ldots, \nu_{m-1}) & \text{if } n < m.
        \end{cases}
    $$
\end{lemme}
We notice that the identity ($i$) states that the partial derivatives of the $m$-th level moments $\mu_n^{(m)}$ are given by elements of $\cal M$ of lower levels, $\mu^{(k - 1)}$ for $k=n,\ldots, m$. In a similar way, we notice that the identity ($ii$) states that the elements of level $m$, namely, $\mu_n^{(m)}$ for $n<m$, are polynomials of degree $n$ in the last variable $\nu_m$ whose coefficients depend on the other elements of $\cal M$ of lower levels, $\mu^{(m - n - 1)}$. 

\begin{proof}
    See Appendix~\ref{proof_murel}.
\end{proof}

The induction identities $(ii)$ are used to compute the analytic expressions of the moments $\mu_n^{(n)}$. The Python code and the expression of the first moments for levels $m\leq 5$ are given in Appendix~\ref{app:code}.  

\begin{coro}
    The centered moments $\mu_n^{(m)}$ of $\cal M$ are homogeneous polynomials of degree $n+1$ for $n=0,\ldots, m$.
\end{coro}
\begin{proof}
    This follows by induction from the identity $(ii)$ in Lemma~\ref{prop:murel}. If $\mu_k^{(m - n - 1)}$ is a homogeneous polynomial of degree $k+1$ for all $k=0,\ldots,n$ and $n=0,\ldots,m-1$, then $\nu_m^{n-k} \mu_k^{(m - n - 1)}$ and, consequently, $\mu_n^{(m)}$ are homogeneous polynomials of degree $n+1$.
\end{proof}

\begin{proposition}
\label{coro:JBrel}
    Let the centered moments $\mu_n^{(m)}$ satisfying the hypotheses of Lemma~\ref{prop:murel}. Then the family of moments $\mu_n^{(m)}$ is a solution of Eq.~\eqref{eqn:JBeqn}.
\end{proposition}
    
\begin{proof}
    See Appendix~\ref{proof_JBrel}.
\end{proof}

\begin{proposition}
\label{coro:JBrelbis}
    The moments $\mu_n^{(m)}$, for $n=0,\ldots, m$, are given by 
\begin{equation}
    \label{eqn:muJB_expr}
    \mu_n^{(m)}(\nu_n,\cdots,\nu_m) = \frac{1}{n+1}\sum_{n\leq i_1,i_2,\ldots, i_{n+1}\leq m \atop i_1+\cdots+i_{n+1} =n(m+1)}\nu_{i_1}\nu_{i_2}\cdots\nu_{i_{n+1}}.
\end{equation}
    
\end{proposition}

\begin{proof}
We prove Eq.~\eqref{eqn:muJB_expr} by induction. The first two cases $n=0$ and $n=1$ are straightforward. We assume Eq.~\eqref{eqn:muJB_expr} holds for $n-1$, i.e.,
$$
    \mu_{n-1}^{(m)}(\nu_{n-1},\dots,\nu_m) = \frac{1}{n} \sum_{n-1\leq i_1,i_2,\ldots, i_{n}\leq m \atop i_1+\dots+i_n = (n-1)(m+1)} \nu_{i_1} \cdots \nu_{i_n}.
$$
We define the candidate formula for $n$ as
$$
    F_n^{(m)}(\nu_n,\dots,\nu_m) = \frac{1}{n+1} \sum_{n\leq i_1,i_2,\ldots, i_{n+1}\leq m \atop i_1+\dots+i_{n+1} = n(m+1)} \nu_{i_1}\cdots\nu_{i_{n+1}}.
$$
For all $k\in \{n,\cdots,m\}$, we compute the partial derivative of $F_n^{(m)}$ with respect to $\nu_k$:
$$
    \frac{\partial F_n^{(m)}}{\partial \nu_k} = \sum_{\substack{i_1+\dots+i_n = n(m+1)-k}} \nu_{i_1}\cdots\nu_{i_n}.
$$
By reindexing 
$$
    j_\ell = i_\ell - (m-k+1), \quad \text{for $\ell=1,\dots,n$},
$$ 
we rewrite the sum as
$$
    \frac{\partial F_n^{(m)}}{\partial \nu_k} = n \, \mu_{n-1}^{(k-1)}(\nu_{m-k+n},\dots,\nu_m),
$$
which matches the identity~\eqref{eqn:dmu_JB}. Since both $F_n^{(m)}$ and $\mu_n^{(m)}$ vanish at the origin, they must coincide: 
$\mu_n^{(m)}(\nu_n,\dots,\nu_m) = F_n^{(m)}(\nu_n,\dots,\nu_m)$. 
\end{proof}

\begin{proposition}
    The parameterizations $\mu_n^{(m)}= \mu_n^{(m)}({\bm\nu})$ are locally invertible for all $m$, and are in fact globally invertible when $m$ is even.
\end{proposition}

\begin{proof}
We proceed by inverting the anti-triangular system starting from $\mu_m^{(m)}$. We notice that the relation $(ii)$ in Lemma~\ref{prop:murel} can be rewritten as
$$
    \mu_n^{(m)} = \nu_n\nu_m^n + \chi_n^{(m)}(\nu_{n+1},\ldots,\nu_m),
$$
where 
$$
    \chi_n^{(m)}(\nu_{n+1},\ldots,\nu_m)= \sum_{k=1}^{n} \binom{n}{k} \nu_m^{n-k} \mu_k^{(m - n - 1)}(\nu_{k+n}, \ldots, \nu_{m-1}),
$$
First, we consider $m$ even. 
For $m \geq 1$, we have:
$$
    \nu_m(\mu_m) = {\rm sgn}(\mu_m^{(m)}) \left((m+1)\vert \mu_m^{(m)}\vert\right)^{1/(m+1)}.
$$
We now move upwards in the anti-triangular system using the induction:
$$
    \nu_n = \left({\rm sgn}(\mu_m^{(m)})\right)^n \left((m+1)\vert \mu_m^{(m)}\vert\right)^{-n/(m+1)}\left( \mu_n^{(m)} -  \overline{\chi}_n^{(m)}(\mu_{n+1}^{(m)},\ldots,\mu_m^{(m)})\right),
$$
where we have inserted the inversion relation from $m$ to $n-1$ such that
$$
    \chi_n^{(m)}(\nu_{n+1},\ldots,\nu_m) = \overline{\chi}_n^{(m)}(\mu_{n+1}^{(m)},\ldots,\mu_m^{(m)}).
$$
Using the same argument the relations $(ii)$ in Lemma~\ref{prop:murel} can also be locally inverted for $m$ odd and $\mu_n^{(m)}>0$. When $m$ is odd and $\mu_n^{(m)}<0$, we consider the following variables:
$$
    \overline{\mu}_n^{(m)}(\nu_{n},\ldots,\nu_m)= (-1)^n \mu_n^{(m)}(\nu_{n},\ldots,\nu_m).
$$
If the polynomials $\overline{\mu}_n^{(m)}$ are solutions of Eq.~\eqref{eqn:JBeqn} with $g$ given by Eq.~\eqref{eqn:JBg}, then the polynomials $\mu_n^{(m)}$ are solution of Eq.~\eqref{eqn:JBeqn} with $-g$. 
We consider the same expressions for the inverse as for $\mu_n^{(m)}$ and replace $\mu_n^{(m)}$ by $(-1)^n \overline{\mu}_n^{(m)}$. 
\end{proof}

\begin{proposition}
    The signature of the bracket~\eqref{eqn:PBJB} is $(\lfloor N/2 \rfloor, \lceil N/2 \rceil)$ or $(\lceil N/2 \rceil, \lfloor N/2 \rfloor)$.
\end{proposition}

\begin{proof}
    This follow from a direct computation of the eigenvalues of $g$ given by Eq.~\eqref{eqn:JBg}. If $\mu_m<0$, the metric is $-g$, so the eigenvalues are opposite to the ones for $\mu_m>0$, where $m=N-2$. 
\end{proof}
For $N$ even, we notice that the signature is the same as for the multi-delta closure of Sec.~\ref{sec:part1delta}.

The cubic polynomial generating the closure is
$$
    \mu_2(\nu_2,\ldots, \nu_{N-2})=\frac{1}{3}\sum_{2\leq i,j,k \leq N-2 \atop i+j+k=2(N-1)}\nu_i \nu_j \nu_k.
$$

\subsection{Four-field Hamiltonian closure of \texorpdfstring{Ref.~\cite{chandre_four-field_2022}}{Ref. [Chandre 2022]}}
\label{sec:part1CS}

We consider a fluid model obtained by using the first $N=4$ fluid moments derived in Ref.~\cite{chandre_four-field_2022}. We recall that the Hamiltonian closure was not obtained from a closed form for the distribution function, but by solving directly the Jacobi identity. We consider the variables $(\rho,u,S_2,S_3)$. The bracket has the form~\eqref{eqn:PBdec} with the \( 2 \times 2 \) matrices \( \alpha \) and \( \beta\) given by
$$
    \alpha = \begin{pmatrix}
    4S_3 & 5S_4 - 9S_2^2 \\
    5S_4 - 9S_2^2 & 6S_5 - 24S_2 S_3
    \end{pmatrix},
$$
and
$$
    \beta = \partial_x \begin{pmatrix}
    2S_3 & 2S_4 - 3S_2^2 \\
    3S_4 - 6S_2^2 & 3S_5 - 12S_2 S_3
    \end{pmatrix}.
$$
In Ref.~\cite{perin_hamiltonian_2015-2}, it was shown that the Jacobi identity leads to the following constraints on the closure functions $S_4$ and $S_5$:
\begin{eqnarray*}
    && 6S_5 = 12S_2 S_3 + 4S_3 \frac{\partial S_4}{\partial S_2} + \left(5S_4 - 9S_2^2\right) \frac{\partial S_4}{\partial S_3},\\
    && \frac{\partial S_5}{\partial S_2} = 4S_3 + \frac{\partial S_4}{\partial S_3} \left( \frac{\partial S_4}{\partial S_2} - 3S_2 \right),\\
    && \frac{\partial S_5}{\partial S_3} = \frac{\partial S_4}{\partial S_2} + \left( \frac{\partial S_4}{\partial S_3} \right)^2.
\end{eqnarray*}
Instead of solving explicitly these coupled partial differential equations, it is wiser to give solutions in parametric form:
$$
    S_n = S_n(\Gamma_2,\Gamma_3),
$$
for $n=2,3,4,5$. Solving for the constraints yields:
\begin{align*}
    S_2 &= \Gamma_2^3 + \Gamma_2(\kappa-\Gamma_2)\Gamma_3^2,\\
    S_3 &= \Gamma_2\Gamma_3(\kappa-\Gamma_2)(3\Gamma_2^2+(\kappa-2\Gamma_2)\Gamma_3^2),\\
    S_4 &= \frac{9\kappa}{5}\Gamma_2^5 + 6\Gamma_2^3(\kappa-\Gamma_2)^2\Gamma_3^2\\
    &+\Gamma_2(\kappa-\Gamma_2)(\kappa^2-3\Gamma_2(\kappa-\Gamma_2))\Gamma_3^4,\\
    S_5 &= 9\kappa\Gamma_2^5(k-\Gamma_2)\Gamma_3 + 10\Gamma_2^3(\kappa-\Gamma_2)^3\Gamma_3^3\\
    &+\Gamma_2(\kappa-\Gamma_2)(\kappa-2\Gamma_2)(\kappa^2-2\kappa\Gamma_2 + 2\Gamma_2^2)\Gamma_3^5,
\end{align*}
where $\kappa$ is a free parameter. We verify that, as expected, $S_n$ is a polynomial of degree $2n$ in the variables $(\Gamma_2,\Gamma_3)$, and these expressions verify the recurrence relation:
$$
    S_{n+1}=\frac{3n}{n+2}S_2 S_{n-1}+\frac{1}{n+2}\left( \frac{\partial S_n}{\partial \Gamma_2}\frac{\partial S_2}{\partial \Gamma_3}+\frac{\partial S_n}{\partial \Gamma_3}\frac{\partial S_2}{\partial \Gamma_2}\right).
$$
Under the change of variables $(\rho,u,S_2,S_3)\mapsto (\rho,u,\Gamma_2,\Gamma_3)$, the bracket becomes of the form~\eqref{eqn:PBdec_nu} with 
$$
    g=\begin{pmatrix}
    0 & 1 \\
    1 & 0
    \end{pmatrix}.
$$
Therefore the signature of the bracket is $(1,1)$. 
Once again, centering with respect to $\mu_1 = \Gamma_2 \Gamma_3$ simplifies the parametric expression of the moments:
\begin{align*}
    \mu_1 &= \Gamma_2 \Gamma_3,\\
    \mu_2 &= \Gamma_2^3 +\kappa\Gamma_2\Gamma_3^2,\\
    \mu_3 &= \kappa \Gamma_2\Gamma_3\left(3\Gamma_2^2+\kappa \Gamma_3^2 \right),\\
    \mu_4 &= \kappa\left(\frac{9\Gamma_2^5}{5}+6\kappa\Gamma_2^3\Gamma_3^2+\kappa^2\Gamma_2\Gamma_3^4 \right),\\
    \mu_5 &= \kappa^2\Gamma_2\Gamma_3\left(9\Gamma_2^4+10\kappa\Gamma_2^2\Gamma_3^2+\kappa^2\Gamma_3^4 \right),
\end{align*}
where we notice that $\mu_n$ is a homogeneous polynomial of degree $n+1$ in the two variables $(\Gamma_2,\Gamma_3)$. The parametric closure satisfies the following recursion relation:
$$
    \mu_{n+1}=\frac{1}{n+2}\left(\frac{\partial \mu_n}{\partial \Gamma_2}\frac{\partial \mu_2}{\partial \Gamma_3}+ \frac{\partial \mu_n}{\partial \Gamma_3}\frac{\partial \mu_2}{\partial \Gamma_2}\right),
$$
which is the same recursion relation as in Sec.~\ref{sec:part1JB} for the case $N=4$.
In addition, we notice that for $\kappa=0$, we recover the case of Sec.~\ref{sec:part1JB}, with in particular, $\mu_n=0$ for $n\geq 4$. However, contrary to the case in Sec.~\ref{sec:part1JB}, Casimir invariants in the variables $\mu_n$ are not expressed in a convenient and simple format, given that they require an inversion of a fourth-order polynomial (see Ref.~\cite{chandre_four-field_2022} for more details).

From the above recursion relation, we deduce that all the information on the closure is contained in the function $\mu_2$, a homogeneous cubic polynomial. For $N=4$, the following generators of the closures are found:
\begin{itemize}
    \item 4-field closure (Sec.~\ref{sec:part1CS}): $\mu_2=\Gamma_2^3+\kappa\Gamma_2\Gamma_3^2$,
    \item Burby's closure (Sec.~\ref{sec:part1JB}): $\mu_2=\Gamma_2^3/3$,
    \item multi-delta closure (Sec.~\ref{sec:part1delta}): $\mu_2=\Gamma_2\Gamma_3^2$.
\end{itemize}
Therefore the closure found in Ref.~\cite{chandre_four-field_2022} is a linear interpolation between the multi-delta closure and Burby's closure~\cite{burby_variable-moment_2023}. As in the waterbag case, this closure has the advantage of having a free parameter $\kappa$ to adjust depending on the kinetic features to be reproduced (see Ref.~\cite{chandre_four-field_2022} for more detail).

Given that the closure is composed of homogeneous polynomials $\mu_n$ of degree $n+1$, the change of variables $(\rho,u,\Gamma_2,\Gamma_3) \mapsto (\rho,u,\mu_1,\mu_2)$ yields a bracket of the form~\eqref{eqn:PBdec} with $\alpha_{nm}({\bm \mu}) = (n+m)\mu_{n+m-1}$ and  $\beta_{nm}({\bm \mu},\partial_x {\bm \mu}) = m\partial_x\mu_{n+m-1}$.

\section*{Conclusions}

In this article, we addressed the problem of finding closure functions for the fluid moments $P_n$, i.e., of the form $P_{n\geq N}=P_{n\geq N}(P_0, \ldots,P_{N-1})$ by first finding suitable centered moments $\mu_{n}$ such that the closure reduces to
$$
    \mu_{n\geq N-1}=\mu_{n\geq N-1}(\mu_1,\ldots,\mu_{N-2}),
$$
i.e., closure functions of $N-2$ variables. These closure functions are constrained by the requirement that the resulting dynamical systems should respect the Hamiltonian structure of the parent model, namely the one-dimensional Vlasov equation. We have shown that all these closure functions, when expressed parametrically with well-chosen normal variables $\nu_k$, are generated by a single moment $\mu_2$.

We found that all known hydrodynamic Hamiltonian closures for the one-dimensional Vlasov equation, as listed in Sec.~\ref{sec:part1}, share a fundamental property: There exists a convenient parameterization of the variables $\mu_{n< N-1}$ and the closure functions  $\mu_{n\geq N-1}$, i.e., 
$$
    \mu_n=\mu_n(\nu_1,\ldots,\nu_{N-2}),
$$
where $\mu_n$ is a polynomial of degree $n+1$ in the normal variables $\nu_k$. It was shown that the bracket in the $\bm\mu$ variables and the change of variables from $(\rho, u, \mu_1,\ldots, \mu_{N-2})$ to $(P_0, P_1,\ldots, P_{N-1})$ involve the specification of $\mu_n$ for $n=1,\ldots,2N-3$ (by including the closure functions $\mu_n$ appearing explicitly in the bracket). The specification of these polynomials of increasing degree from 2 to $2N-2$ determines completely the hydrodynamic Hamiltonian closure. We have demonstrated that a given closure is uniquely determined by a single moment, $\mu_2$. For all known closures, this moment is expressed as a cubic polynomial in the $N-2$ variables $\nu_k$. Expressed in the normal variables, the Hamiltonian is given by:
$$
    H[\rho, u, \nu_1,\ldots, \nu_{N-2}]= \frac{1}{2}\int \left[\rho u^2 + \rho^3 \left(\mu_2(\nu_1,\ldots, \nu_{N-2}) - \mu_1(\nu_1,\ldots, \nu_{N-2})^2 \right) + E^2 \right]\, {\rm d}x,
$$
where $\mu_1$ is a homogeneous quadratic polynomial, defined by Eq.~\eqref{eqn:mu1}. Within this framework, the internal energy density $\rho^2 (\mu_2 - \mu_1^2)/2$ is given by the difference between a cubic polynomial in the $N-2$ variables $\nu_k$ and the square of a homogeneous quadratic polynomial in $\nu_k$. Consequently, the specification of the polynomial $\mu_2$ fully characterizes the Hamiltonian system, i.e., Poisson bracket and Hamiltonian. Despite defining a Hamiltonian system, any cubic polynomial does not necessarily lead to a model linked to a reduction of the Vlasov equation. However, following up on a suggestion made in Ref.~\cite{burby_variable-moment_2023}, this family of solutions can be a good starting point for data-driven closures, using sparse identification of non-linear dynamics (SINdy~\cite{Brunton_Kutz_2022}) or sparse physics-informed discovery of empirical relations (SPIDER~\cite{Gurevich_Golden_Reinbold_Grigoriev_2024}).

In addition to providing natural variables in which the closures get the simple polynomial form, these normal variables $\nu_k$ greatly simplify the expression of the Poisson bracket:
$$
    \{F,G\}=\sum_k \epsilon_k  \int \partial_x \frac{\delta F}{\delta \nu_k} \frac{\delta G}{\delta \nu_k} \ {\rm d}x,
$$
which is very convenient from a computational point of view since it allows for a split of the Poisson bracket as a direct sum of Poisson brackets. Each Poisson bracket individually preserves all the Casimir invariants. The implementation of a split-integrator (by splitting the Poisson bracket, not necessarily the Hamiltonian) ensures the conservation of these quantities.

\begin{acknowledgments}
The authors would like to thank useful conversations with Josh Burby, Phil Morrison, Maxime Perin and Brad Shadwick. This work has been carried out within the framework of the EUROfusion Consortium, funded by the European Union via the Euratom Research and Training Programme (Grant Agreement No 101052200 — EUROfusion). Views and opinions expressed are however those of the author(s) only and do not necessarily reflect those of the European Union or the European Commission. Neither the European Union nor the European Commission can be held responsible for them. 
\end{acknowledgments}

\section*{Data Availability Statement} 
This manuscript describes an entirely theoretical approach and does not report on any primary data. Data sharing is therefore not applicable.

\section*{Author Contribution Statement}
{\bf R. Oufar}: Methodology, Formal analysis, Writing - Original Draft.
{\bf C. Chandre}: Conceptualization, Methodology, Formal analysis, Software, Writing - Original Draft.

\appendix

\section{Expression of the bracket in the \texorpdfstring{${\bm \mu}$}{mu}-variables}
\label{app:PBmu}

We perform the change of variables $(P_0,P_1,\ldots,P_{N-1})\mapsto (\rho,\psi,\mu_1,\ldots,\mu_{N-2})$, given by 
\begin{equation}
    \label{eqn:Pvsmu}
    P_n=\sum_{k=0}^n\binom{n}{k}\mu_k\rho^{k+1}\psi^{n-k},
\end{equation}
obtained from
$$
    P_n = \int (v-\psi+\psi)^n f(t,x,v)\ {\rm d}v.
$$
The relation~\eqref{eqn:Pvsmu} is inverted:
\begin{equation}
    \label{eqn:muvsP}
    \mu_n=\frac{1}{\nu_0^{n+1}}\sum_{k=0}^n\binom{n}{k}P_k(-\nu_1)^{n-k}.
\end{equation}
Following Lemma~\ref{lem:albe}, the moments $P_n$ satisfy the following relations in the variables $\bm\nu$:
\begin{subequations}
\begin{eqnarray}
    && (n+m)P_{n+m-1}=\frac{\partial P_n}{\partial\nu_k} g_{kl} \frac{\partial P_m}{\partial\nu_l},\label{eqn:nmP}\\
    && n\partial_x P_{n+m-1}=\partial_x \frac{\partial P_n}{\partial\nu_k} g_{kl} \frac{\partial P_m}{\partial\nu_l},\label{eqn:dxP}
\end{eqnarray}
\end{subequations}
for $n,m=0,\ldots N-1$, and where $g$ is the constant metric.
We assume that $g$ takes the block form
$$
    g = \begin{pmatrix}
        g_0 & \mathbb{O} \\
        \mathbb{O}^t & \overline{g}
    \end{pmatrix},
$$
where $\mathbb{O}$ is the $2\times (N-2)$ zero matrix, $\overline{g}$ is a $(N-2)\times(N-2)$ constant matrix, and 
$$
    g_0 = \begin{pmatrix}
        0 & 1\\
        1 & 0
    \end{pmatrix},
$$
to account for the partial decoupling in the variables $\nu_0=\rho$ and $\nu_1=\psi$, i.e., $\rho$ is canonically conjugate to $\psi$. Using Eq.~\eqref{eqn:nmP} for $m=0$, we have
\begin{equation}
\label{eqn:appPn1}
    \frac{\partial P_n}{\partial \nu_1}=n P_{n-1}.
\end{equation}
By differentiating Eq.~\eqref{eqn:Pvsmu} with respect to $\nu_1$, and using the equation above, we show that
$$
    \frac{\partial \mu_n}{\partial \nu_1} =0.
$$
From Eq.~\eqref{eqn:dxP} with $n=1$, we have
$$
    \left(g_{kl}\frac{\partial^2 P_1}{\partial\nu_k\nu_l}-\delta_{kl}\right) \frac{\partial P_m}{\partial {\nu_l}}=0,
$$
for all $m=0,\ldots,N-1$. Therefore, $P_1$ is a quadratic polynomial in $\bm\nu$:
$$
    P_1= \frac{1}{2}{\bm\nu}\cdot g^{-1}{\bm\nu}+ {\bf h}\cdot {\bm\nu}+ k,
$$
where ${\bf h}$ and $k$ are constants. Using a gauge on ${\bm\nu}$ (arbitrary translation in ${\bm\nu}$ by a constant vector), we can always assume that ${\bf h}={\bf 0}$. From Eq.~\eqref{eqn:nmP} with $n=m=1$, we deduce that $k=0$. As a consequence, for all closures, we can always find variables ${\bm\nu}$ such that $P_1$ is a homogeneous quadratic polynomial in these variables, i.e.,
$$
    P_1= \frac{1}{2}{\bm\nu}\cdot g^{-1}{\bm\nu}.
$$
Ensuring that $P_1$ is a homogeneous quadratic polynomial is fixing the gauge on ${\bm\nu}$. 
Since $P_1=\nu_0\nu_1+\nu_0^2\mu_1$, the centered moment $\mu_1$ can be rewritten as
$$
    \mu_1=\frac{1}{2}\frac{\nu_k}{\nu_0} \left(\bar{g}^{-1}\right)_{kl}\frac{\nu_l}{\nu_0},
$$
where the implicit summation over $k,l$ is from 2 to $N-1$. Next, from Eq.\eqref{eqn:muvsP}, we compute the partial derivative of $\mu_n$ with respect to $\nu_0$:
\begin{equation}
    \label{eqn:appmu0}
    \nu_0\frac{\partial \mu_n}{\partial \nu_0}=-(n+1)\mu_n + \frac{1}{\nu_0^{n+1}}\sum_{k=0}^n\binom{n}{k}\nu_0\frac{\partial P_k}{\partial \nu_0}(-\nu_1)^{n-k},
\end{equation}
and using Eq.~\eqref{eqn:nmP} with $m=1$:
$$
    \nu_0\frac{\partial P_k}{\partial \nu_0}=(k+1)P_k -\nu_1\frac{\partial P_k}{\partial \nu_1} -\nu_l\frac{\partial P_k}{\partial \nu_l},
$$
where the implicit summation over $l$ is from 2 to $N-1$. The operator $\nu_l{\partial}/{\partial \nu_l}$ commutes with the sum over $k$ in the right-hand side of Eq.~\eqref{eqn:appmu0}.  As for the operator $\nu_1{\partial}/{\partial \nu_1}$, we use Eq.~\eqref{eqn:appPn1}, to show that the term of the right-hand side of Eq.~\eqref{eqn:appmu0} involving $\nu_1{\partial}/{\partial \nu_1}$ cancels out with the term involving $(k+1)P_k$. Consequently, we have
$$
    \nu_0\frac{\partial \mu_n}{\partial \nu_0}=-\nu_k \frac{\partial \mu_n}{\partial \nu_k},
$$
where the implicit sum over $k$ is from 2 to $N-1$. Thus leads to the solution:
$$
    \mu_n(\nu_0, \nu_2,\ldots, \nu_{N-1})=\overline{\mu}_n \left(\frac{\nu_2}{\nu_0},\ldots,\frac{\nu_{N-1}}{\nu_0}\right).
$$
In other words, the centered moments $\mu_n$ for $n\geq 1$ are functions of the rescaled normal variables $\nu_k/\nu_0$ for $k=2,\ldots,N-1$. Next we compute 
$$
    \nu_0^2 \frac{\partial \mu_n}{\partial \nu_k}\overline{g}_{kl}\frac{\partial \mu_m}{\partial \nu_l}=\frac{1}{\nu_0^{n+m}}\sum_{q=0}^n\sum_{p=0}^m \binom{n}{q}\binom{m}{p}\frac{\partial P_q}{\partial\nu_k} \overline{g}_{kl} \frac{\partial P_p}{\partial\nu_l}(-\nu_1)^{n+m-q-p},
$$
for $n,m\geq 2$ and where the implicit sums over $k$ and $l$ are from 2 to $N-1$.
Using Eq.~\eqref{eqn:nmP}, we have
$$
\frac{\partial P_n}{\partial\nu_k} g_{kl} \frac{\partial P_m}{\partial\nu_l}= \frac{\partial P_n}{\partial\nu_k} \overline{g}_{kl} \frac{\partial P_m}{\partial\nu_l} + mP_{m-1} \frac{\partial P_n}{\partial\nu_0} + nP_{n-1} \frac{\partial P_m}{\partial\nu_0},
$$
where we have used Eq.~\eqref{eqn:appPn1}. Using the Vandermonde identity, 
$$
    \sum_{q=0}^s \binom{n}{q}\binom{m}{s-q}=\binom{n+m}{s},
$$
the expression simplifies to
$$
    \frac{\partial \bar\mu_n}{\partial \bar\nu_k}\overline{g}_{kl}\frac{\partial \bar\mu_m}{\partial \bar\nu_l} = (n+m) \bar\mu_{n+m-1}-n\bar\nu_{n-1}\gamma_m-m\bar\nu_{m-1}\gamma_n,
$$
where $\gamma_n$ is given by Eq.~\eqref{eqn:gamma} and $\bar{\bm\nu}={\bm\nu}/\nu_0$. Using similar calculations, we have
$$
    \partial_x\frac{\partial \bar\mu_n}{\partial \bar\nu_k}\overline{g}_{kl}\frac{\partial \bar\mu_m}{\partial \bar\nu_l} = n \partial_x \bar\mu_{n+m-1}-n\gamma_m \partial_x \nu_{n-1}-m\nu_{m-1}\partial_x \gamma_n.
$$
As a consequence, the bracket in the $\bar{\bm\mu}$ variables is a hydrodynamic bracket of the form~\eqref{eqn:PBdec} with elements
\begin{eqnarray}
    && \bar{\alpha}_{nm}=(n+m) \bar\mu_{n+m-1}-n\bar\nu_{n-1}\gamma_m-m\bar\nu_{m-1}\gamma_n,\\
    && \bar{\beta}_{nm}=n \partial_x \bar\mu_{n+m-1}-n\gamma_m \partial_x \nu_{n-1}-m\nu_{m-1}\partial_x \gamma_n.
\end{eqnarray}

\section{Proofs of statements in \texorpdfstring{Sec.~\ref{sec:part1JB}}{Section 1}}

\subsection{Proof of Lemma~\ref{prop:murel}}
\label{proof_murel}

Since we are looking for solutions of the type defined in Lemma~\ref{prop:munu}, $\mu_n^{(m)}$ only depends on $m+1-n$ variables, e.g., $(\nu_n, \ldots,\nu_m)$, only the mention of the index of the last variable is needed, e.g., without ambiguity, $\mu_n^{(m)}(\ldots,x_l)$ is thus a shorthand for $\mu_n^{(m)}(x_{n-m+l},\ldots,x_l)$. Consequently, the identity ($i$) can be rewritten as
\begin{equation}
\label{eqn:derivk}
    \partial_k \mu_n^{(m)}(\ldots,x_{m+1-n})= n \mu_{n-1}^{(n+k-2)}(\ldots,x_{m+1-n}),
\end{equation}
where $\partial_k$ is the partial derivative of $\mu_n^{(m)}$ with respect to the $k$-th variable (here $x_k$). In particular, we notice that the derivative with respect to the last variable, i.e., for $k=m+1-n$, the variable dependencies do not change between the left and the right-hand side, and
$$
    \partial_{m+1-n} \mu_n^{(m)}(x_1,\ldots,x_{m+1-n})= n \mu_{n-1}^{(m-1)}(x_1,\ldots,x_{m+1-n}).
$$

$\bullet$ ($i$) implies ($ii$):  We assume that $(ii)$ is satisfied at all lower levels than $m$ (i.e., for $\mu_n^{(k)}$ with $k<m$), and we show that $(ii)$ is satisfied at level $m$ using $(i)$. \\
From Eq.~\eqref{eqn:derivk}, we compute the $k$-th partial derivative of $\mu_n^{(m)}$ with respect to $\nu_m$:
$$
    \frac{\partial^k}{\partial \nu_m^k}\left[\mu_n^{(m)}(\ldots,\nu_m)\right]=n(n-1)\cdots(n-k+1)\mu_{n-k}^{(m-k)}(\ldots,\nu_m),
$$
for $1\leq k\leq n$, and 
\begin{eqnarray*}
    && \frac{\partial^{n+1}}{\partial \nu_m^{n+1}}\left[\mu_n^{(m)}(\ldots,\nu_m)\right]=n! \delta_{mn},\\
    && \frac{\partial^k}{\partial \nu_m^k}\left[\mu_n^{(m)}(\ldots,\nu_m)\right]=0,
\end{eqnarray*}
for $k>n+1$, where we have used $\mu_0^{(l)}(\ldots,x_l)=x_0$. Consequently, $\mu_n^{(m)}$ is a polynomial in the variable $\nu_m$ of degree $n$ if $n<m$, and degree $m+1$ if $n=m$. By writing the polynomial using a series expansion around $\nu_m=0$, we have for $m>n$,
\begin{eqnarray*}
    \mu_n^{(m)}(\ldots,\nu_m)&=&\sum_{k=0}^n \frac{\nu_m^{n-k}}{(n-k)!} \frac{\partial^{n-k}\mu_n^{(m)}}{\partial \nu_m^{n-k}}(\ldots,\nu_{m-1},0),\\
    &=& \sum_{k=0}^n\binom{n}{k} \nu_m^{n-k} \mu_k^{(m-n+k)}(\ldots,\nu_{m-1},0).
\end{eqnarray*} 
From $(ii)$ at lower levels, i.e., for $k<n$, we have by taking $\nu_m=0$:
$$
    \mu_k^{(m-n+k)}(\ldots,\nu_{m-1},0)=\mu_k^{(m-n-1)}(\ldots,\nu_{m-1}).
$$
In order to prove that this identity is also true for $k=n$, we differentiate with respect to $\nu_l$ for $l=n,\ldots,m-1$, and show that
$$
    \frac{\partial}{\partial \nu_l}\mu_n^{(m)}(\ldots,\nu_{m-1},0)=\frac{\partial}{\partial \nu_l}\mu_n^{(m-n-1)}(\ldots,\nu_{m-1}),
$$
since this identity is equivalent to
$$
    \mu_{n-1}^{(l-1)}(\ldots,\nu_{m-1},0)=\mu_{n-1}^{(l-n-1)}(\ldots,\nu_{m-1}),
$$
which is true for all $l=n,\ldots,m-1$, given $(ii)$ at lower levels. Assuming that $\mu_n^{(m)}({\bf 0})=0$, we deduce that 
$$
    \mu_n^{(m)}(\ldots,\nu_{m-1},0)=\mu_n^{(m-n-1)}(\ldots,\nu_{m-1}).
$$
The identity $(ii)$ follows for $n<m$. As for $n=m$, since
$$
    \frac{\partial}{\partial \nu_m}\left[\mu_m^{(m)}(\nu_m)\right] = m\mu_{m-1}^{(m-1)}(\nu_{m}),
$$
and $\nu_0^{(0)}(\nu_m)=0$, we deduce that $\nu_m^{(m)}(\nu_{m})=\nu_m^{m+1}/(m+1)$ by induction.

$\bullet$ ($ii$) implies ($i$): \\
We assume that ($ii$) is satisfied for $\mu_n^{(p)}$ where $p=1,\ldots,m-1$ (and $n=1,\ldots,p$). For $k=m$, the partial derivative of $\mu_n^{(m)}$ with respect to $\nu_m$ becomes
\begin{eqnarray*}
    \frac{\partial \mu_n^{(m)}}{\partial \nu_m}(\ldots,\nu_m)&=& \sum_{k=0}^{n-1} \binom{n}{k} (n-k) \nu_m^{n-k-1} \mu_k^{(m - n - 1)}(\ldots, \nu_{m-1}),\\
    &=& n  \sum_{k=0}^{n-1} \binom{n-1}{k} \nu_m^{n-k-1} \mu_k^{(m - n - 1)}( \ldots, \nu_{m-1}),\\
    &=& n\mu_{n-1}^{(m-1)}(\ldots,\nu_m),
\end{eqnarray*}
where we have used $(n-k)\binom{n}{k}=n\binom{n-1}{k}$.
Next, for $k<m$, the partial derivative of $\mu_n^{(m)}$ with respect to $\nu_l$ is given by
\begin{eqnarray*}
    \frac{\partial \mu_n^{(m)}}{\partial \nu_l}(\ldots,\nu_m)&=& \sum_{k=0}^{n} \binom{n}{k} \nu_m^{n-k} \frac{\partial}{\partial \nu_l}\left[\mu_k^{(m - n - 1)}( \ldots, \nu_{m-1})\right],\\
    &=&  \sum_{k=0}^{n} k \binom{n}{k} \nu_m^{n-k} \mu_{k-1}^{(l - n - 1)}( \ldots, \nu_{m-1}),\\
    &=& n\mu_{n-1}^{(l-1)}(\ldots,\nu_m),
\end{eqnarray*}
where we have used $(i)$ at strictly lower orders in the form of Eq.~\eqref{eqn:derivk}, and $k\binom{n}{k}=n \binom{n-1}{k-1}$.

\subsection{Proof of Proposition~\ref{coro:JBrel}}
\label{proof_JBrel}

Given the functional dependencies of $\mu_k^{(m)}$, we already know that these identities are satisfied for $k+l>m+1$ at any level $m$ (see Lemma~\ref{prop:munu}). 

    $\bullet$ Proof of Eq.~\eqref{eqn:JBeqn_a}:

From $(i)$ of Lemma~\ref{prop:murel}, we obtain
$$
    \frac{\partial \mu_k^{(m)}}{\partial \nu_j}\frac{\partial \mu_l^{(m)}}{\partial \nu_{m+1-j}}=k l \sum_{j=k}^{m-l+1} \mu_{k-1}^{(j-1)}(\ldots,\nu_m) \mu_{l-1}^{(m-j)}(\ldots,\nu_m).
$$
We extend the summation to all $j$ by assuming that $\mu_k^{(j)}=0$ if $k>j$ or $j<0$. Below the index $j$ is implicitly summed. By using $(ii)$ of Lemma~\ref{prop:murel}, the summation is further expanded to
\begin{eqnarray*}
    \frac{\partial \mu_k^{(m)}}{\partial \nu_j}\frac{\partial \mu_l^{(m)}}{\partial \nu_{m+1-j}}&=&k l \sum_{k',l'} \binom{k-1}{k'}\nu_m^{k-1-k'}\mu_{k'}^{(j-k-1)}(\ldots,\nu_{m-1})\binom{l-1}{l'}\nu_m^{l-1-l'}\mu_{l'}^{(m-j-l)}(\ldots,\nu_{m-1}),\\
    &=& k l \sum_{k',l'}\binom{k-1}{k'}\binom{l-1}{l'}\nu_m^{k-1-k'}\nu_m^{l-1-l'} \mu_{k'}^{(j-k-1)}(\ldots,\nu_{m-1}) \mu_{l'}^{(m-j-l)}(\ldots,\nu_{m-1}).
\end{eqnarray*}
Next we use $(ii)$ at levels strictly smaller than $m$:
$$
    \mu_{k'}^{(j-k-1)} \mu_{l'}^{(m-j-l)}=\frac{k'+l'+2}{(k'+1)(l'+1)}\mu_{k'+l'+1}^{(m-k-l)}.
$$
Using $k\binom{k-1}{k'}=(k'+1)\binom{k}{k'+1}$, we have
$$
    \frac{\partial \mu_k^{(m)}}{\partial \nu_j}\frac{\partial \mu_l^{(m)}}{\partial \nu_{m+1-j}}=\sum_{k'}\sum_s \binom{k}{k'+1}\binom{l}{s-k'+1}\nu_m^{k+l-s-2}(s+2)\mu_{s+1}^{m-k-l}.
$$
We simplify the product of the binomial coefficients using the Vandermonde identity
$$
\sum_{k'}\binom{k}{k'+1}\binom{l}{s-k'+1}=\binom{k+l}{s+2}, 
$$
to get
$$
    \frac{\partial \mu_k^{(m)}}{\partial \nu_j}\frac{\partial \mu_l^{(m)}}{\partial \nu_{m+1-j}} = \sum_{s} (s+2) \binom{k+l}{s+2}\nu_m^{k+l-s-2}\mu_{s+1}^{(m-k-l)},
$$
which is equal to $(k+l)\mu_{k+l-1}^{(m)}$ since $(s+2) \binom{k+l}{s+2}=(k+l) \binom{k+l-1}{s+1}$. 

To initiate the induction, we prove that the identity~\eqref{eqn:JBeqn_a} is trivially satisfied at level $m=1$ where there is only one non-zero polynomial, $\nu_1^{(1)}(\nu_1)=\nu_1^2/2$. 

    $\bullet$ Proof of Eq.~\eqref{eqn:JBeqn_b}:

This proof follows the same lines as for Eq.~\eqref{eqn:JBeqn_a}. 
Proving Eq.~\eqref{eqn:JBeqn_b} at level $m$ is equivalent to proving:
$$
    \frac{\partial^2 \mu_k^{(m)}}{\partial \nu_n\partial \nu_j}\frac{\partial \mu_l^{(m)}}{\partial \nu_{m+1-j}}=k\frac{\partial \mu_{k+l-1}^{(m)}}{\partial \nu_{n}}.
$$
We make use of the identities $(i)$ of Lemma~\ref{prop:murel}, from which we deduce:
\begin{eqnarray*}
    && \frac{\partial}{\partial\nu_n}\left[\mu_k^{(m)}(\ldots,\nu_m) \right]=k\mu_{k-1}^{(n-1)}(\ldots,\nu_m),\\
    && \frac{\partial^2}{\partial\nu_n\partial \nu_j}\left[\mu_k^{(m)}(\ldots,\nu_m) \right]=k(k-1)\mu_{k-2}^{(n-2+j-m)}(\ldots,\nu_m),\\
    && \frac{\partial}{\partial\nu_{m+1-j}}\left[\mu_l^{(m)}(\ldots,\nu_m) \right]=l\mu_{l-1}^{(m-j)}(\ldots,\nu_m).
\end{eqnarray*}
Furthermore, we use $(ii)$ from Lemma~\ref{prop:murel}, i.e., 
$$
    \mu_{k-2}^{(n-2+j-m)} \mu_{l-1}^{(m-j)}=\frac{k+l-1}{(k-1)l}\mu_{k+l-2}^{(n-1)}.
$$
Consequently,
\begin{eqnarray*}
    \frac{\partial^2 \mu_k^{(m)}}{\partial \nu_n\partial \nu_j}\frac{\partial \mu_l^{(m)}}{\partial \nu_{m+1-j}}&=&k(k+l-1)\mu_{k+l-2}^{(n-1)}, \\
    &=& k\frac{\partial}{\partial\nu_n}\left[\mu_{k+l-1}^{(m)}(\ldots,\nu_m) \right],
\end{eqnarray*}
where we have used $(i)$ of Lemma~\ref{prop:murel}. 

To initiate the recursion, we also consider the lowest level $m=1$ where the identity~\eqref{eqn:JBeqn_b} is trivially satisfied.

\section{Python code for the computation of Casimir invariants associated with Burby's closure~\texorpdfstring{\cite{burby_variable-moment_2023}}{[Burby, 2023]}}
\label{app:code}

\begin{lstlisting}

import sympy as sp

def compute_mu(n, l, *args):
    if len(args) == 0 or l < 0 or l > n:
        return 0
    if len(args) != n - l + 1:
        raise ValueError(f"Dimension of mu({n},{l}) should be {n - l + 1} 
                            and not {len(args)}")
    mu = args[0] * args[-1]**l if len(args)>=2 else args[0]**(l + 1) / (l + 1)
    for k in range(1, l + 1):
        mu += sp.binomial(l, k) * args[-1]**(l - k)\ 
                * compute_mu(n - l - 1, k, *args[k:-1])
    return mu

def compute_vec_mu(n, *args):
    mu = [None] * n
    for l in range(n):
        mu[l] = sp.simplify(compute_mu(n, l + 1, *args[l:]))
    return mu

def inverse_mu(n, *args):
    x_ = sp.Matrix(sp.symbols(f"y1:{n+1}"))
    x_[-1] = ((n + 1) * args[-1])**sp.Rational(1, n + 1)
    for l in range(2, n + 1):
        x_[-l] = args[-l] - sp.simplify(compute_mu(n, n - l + 1, *x_[-l:])\ 
                                        - x_[-l] * x_[-1]**(n - l + 1)) 
        x_[-l] *= x_[-1]**(l - 1 - n)
    return [eq for row in sp.simplify(x_).tolist() for eq in row]

def print_var(vars, **kargs):
    output, name = kargs.get("output", "normal"), kargs.get("name", "x")
    vars = [vars] if isinstance(vars, sp.Basic) else vars
    for var in vars:
        if output == "normal":
            print(var)
        elif output == "latex":
            print(sp.latex(var).replace("x", name))

m = 3
x = sp.symbols(f"x1:{m + 1}")
print_var(compute_vec_mu(m, *x))
print_var(compute_vec_mu(m, *x), output="latex", name=r"\nu")
print_var(inverse_mu(m, *x))
print_var(inverse_mu(m, *x), output="latex", name=r"\mu")
\end{lstlisting}

The first levels $m$ are given below:
\begin{itemize}
    \item $m=2$:
        \begin{eqnarray*}
            && \mu_1^{(2)} = \nu_{1} \nu_{2},\\
            && \mu_2^{(2)} = \frac{\nu_{2}^{3}}{3}.
        \end{eqnarray*}
    \item $m=3$:
        \begin{eqnarray*}
            && \mu_1^{(3)} = \nu_{1} \nu_{3} + \frac{\nu_{2}^{2}}{2},\\
            && \mu_2^{(3)} = \nu_{2} \nu_{3}^{2},\\
            && \mu_3^{(3)} = \frac{\nu_{3}^{4}}{4}.
        \end{eqnarray*}
    \item $m=4$:
        \begin{eqnarray*}
            && \mu_1^{(4)} = \nu_{1} \nu_{4} + \nu_{2} \nu_{3},\\
            && \mu_2^{(4)} =  \nu_{2} \nu_{4}^2 + \nu_{3}^{2}\nu_{4},\\
            && \mu_3^{(4)} = \nu_{3} \nu_{4}^{3},\\
            && \mu_4^{(4)} = \frac{\nu_{4}^{5}}{5}.
        \end{eqnarray*}
    \item $m=5$:
        \begin{eqnarray*}
            && \mu_1^{(5)} = \nu_{1} \nu_{5} + \nu_{2} \nu_{4} + \frac{\nu_{3}^{2}}{2},\\
            && \mu_2^{(5)} = \nu_{2} \nu_{5}^{2} + 2 \nu_{3} \nu_{4} \nu_{5} + \frac{\nu_{4}^{3}}{3},\\
            && \mu_3^{(5)} =  \nu_{3} \nu_{5}^3 + \frac{3}{2} \nu_{4}^{2}\nu_{5}^{2},\\
            && \mu_4^{(5)} = \nu_{4} \nu_{5}^{4},\\
            && \mu_5^{(5)} = \frac{\nu_{5}^{6}}{6}.
        \end{eqnarray*}
\end{itemize}
Their inversion provide the explicit expression of the density of the Casimir invariants $C=\int \rho \nu_n({\bm\mu}) {\rm d} x$: 
\begin{itemize}
    \item $m=2$:
        \begin{eqnarray*}
            && \nu_1^{(2)} = \frac{3^{\frac{2}{3}} \mu_{1}}{3 \sqrt[3]{\mu_{2}}},\\
            && \nu_2^{(2)} = \sqrt[3]{3} \sqrt[3]{\mu_{2}}.
        \end{eqnarray*}
    \item $m=3$:
        \begin{eqnarray*}
            && \nu_1^{(3)} = \frac{\sqrt{2} \left(8 \mu_{1} \mu_{3} - \mu_{2}^{2}\right)}{16 \mu_{3}^{\frac{5}{4}}},\\
            && \nu_2^{(3)} = \frac{\mu_{2}}{2 \sqrt{\mu_{3}}},\\
            && \nu_3^{(3)} = \sqrt{2} \sqrt[4]{\mu_{3}}.
        \end{eqnarray*}
    \item $m=4$:
        \begin{eqnarray*}
            && \nu_1^{(4)} = \frac{5^{\frac{4}{5}} \left(25 \mu_{1} \mu_{4}^{2} - \mu_{3} \left(5 \mu_{2} \mu_{4} - \mu_{3}^{2}\right)\right)}{125 \mu_{4}^{\frac{11}{5}}},\\
            && \nu_2^{(4)} = \frac{5^{\frac{3}{5}} \left(5 \mu_{2} \mu_{4} - \mu_{3}^{2}\right)}{25 \mu_{4}^{\frac{7}{5}}} ,\\
            && \nu_3^{(4)} = \frac{5^{\frac{2}{5}} \mu_{3}}{5 \mu_{4}^{\frac{3}{5}}},\\
            && \nu_4^{(4)} = \sqrt[5]{5} \sqrt[5]{\mu_{4}}.
        \end{eqnarray*}
    \item $m=5$:
        \begin{eqnarray*}
            && \nu_1^{(5)} = \frac{6^{\frac{5}{6}} \left(5184 \mu_{1} \mu_{5}^{3} - 864 \mu_{2} \mu_{4} \mu_{5}^{2} - 432 \mu_{3}^{2} \mu_{5}^{2} + 504 \mu_{3} \mu_{4}^{2} \mu_{5} - 91 \mu_{4}^{4}\right)}{31104 \mu_{5}^{\frac{19}{6}}},\\
            && \nu_2^{(5)} = \frac{6^{\frac{2}{3}} \left(27 \mu_{2} \mu_{5}^{2} - \mu_{4}  \left(9 \mu_{3} \mu_{5} - 2 \mu_{4}^{2}\right)\right)}{162 \mu_{5}^{\frac{7}{3}}},\\
            && \nu_3^{(5)} =  \frac{\sqrt{6} \left(4 \mu_{3} \mu_{5} - \mu_{4}^{2}\right)}{24 \mu_{5}^{\frac{3}{2}}},\\
            && \nu_4^{(5)} = \frac{\sqrt[3]{6} \mu_{4}}{6 \mu_{5}^{\frac{2}{3}}},\\
            && \nu_5^{(5)} = \sqrt[6]{6} \sqrt[6]{\mu_{5}}.
        \end{eqnarray*}
\end{itemize}


%

\end{document}